\documentclass[12pt]{article}
\usepackage{amssymb}
\usepackage{amsfonts}
\usepackage{amsmath}
\usepackage{array}
\usepackage{bm}
\usepackage{booktabs}
\usepackage{xcolor}
\usepackage[nohead]{geometry}
\usepackage[bottom]{footmisc}
\usepackage{endnotes}
\usepackage{graphicx}
\usepackage{multirow}
\usepackage{rotating}
\usepackage{setspace}
\usepackage{siunitx}
\usepackage[normalem]{ulem}
\usepackage{comment}
\usepackage{subfig}
\usepackage{lineno}
\usepackage[round, comma, sort]{natbib}
\usepackage{etoolbox}
\usepackage[colorlinks=true, citecolor=cyan, linkcolor=blue]{hyperref}

\newtheorem{definition}{Definition}

\newtheorem{property}{Property}
\newtheorem{proposition}{Proposition}

\newenvironment{proof}[1][Proof]{\noindent\textbf{#1.} }{\ \rule{0.5em}{0.5em}}
\makeatletter
\def\@biblabel#1{\hspace*{-\labelsep}}
\patchcmd{\NAT@citex}
  {\@citea\NAT@hyper@{%
     \NAT@nmfmt{\NAT@nm}%
     \hyper@natlinkbreak{\NAT@aysep\NAT@spacechar}{\@citeb\@extra@b@citeb}%
     \NAT@date}}
  {\@citea\NAT@nmfmt{\NAT@nm}%
   \NAT@aysep\NAT@spacechar\NAT@hyper@{\NAT@date}}{}{}
\patchcmd{\NAT@citex}
  {\@citea\NAT@hyper@{%
     \NAT@nmfmt{\NAT@nm}%
     \hyper@natlinkbreak{\NAT@spacechar\NAT@@open\if*#1*\else#1\NAT@spacechar\fi}%
       {\@citeb\@extra@b@citeb}%
     \NAT@date}}
  {\@citea\NAT@nmfmt{\NAT@nm}%
   \NAT@spacechar\NAT@@open\if*#1*\else#1\NAT@spacechar\fi\NAT@hyper@{\NAT@date}}
  {}{}

\makeatother
\begin{document}
\title{Laplacian Eigenvector Centrality%
\thanks{
The authors thank Masaki Aoyagi, Francis Bloch, In{\'a}cio B{\'o}, Makoto Hanazono, Michihiro Kandori, and Yu Zhou for their helpful comments on an earlier version. Tamura gratefully acknowledges financial support from JSPS Grant-in-Aid for Scientific Research No. 24K04780.}}
\author{Koya Shimono%
\thanks{Graduate School of Economics, Nagoya University. E-mail: shimono.koya.x0@s.mail.nagoya-u.ac.jp}%
 \and Wataru Tamura%
\thanks{Graduate School of Economics, Nagoya University. E-mail: wtr.tamura@gmail.com.
}%
}
\maketitle
\sloppy%

\onehalfspacing
\begin{abstract}\noindent
Networks significantly influence social, economic, and organizational outcomes, with centrality measures serving as crucial tools to capture the importance of individual nodes. This paper introduces Laplacian Eigenvector Centrality (LEC), a novel framework for network analysis based on spectral graph theory and the eigendecomposition of the Laplacian matrix.
A distinctive feature of LEC is its adjustable parameter, the LEC order, which enables researchers to control and assess the scope of centrality measurement using the Laplacian spectrum.
Using random graph models, LEC demonstrates robustness and scalability across diverse network structures. We connect LEC to equilibrium responses to external shocks in an economic model, showing how LEC quantifies agents' roles in attenuating shocks and facilitating coordinated responses through quadratic optimization.
Finally, we apply LEC to the study of microfinance diffusion, illustrating how it complements classical centrality measures, such as eigenvector and Katz-Bonacich centralities, by capturing distinctive aspects of node positions within the network.
\vspace{4mm}
\end{abstract}

\onehalfspacing
\clearpage
\section{Introduction}\label{sec:intro}

Networks play a pivotal role in shaping interactions and outcomes across social, economic, and organizational contexts. To understand the influence of network structure on individual behavior and collective dynamics, centrality measures, which quantify the importance of nodes within a network, have been widely studied and applied, reflecting diverse concepts such as influence, power, or prestige (\citealp{jackson2017economic}). However, existing measures are often constrained by specific assumptions or methodological rigidity, limiting their adaptability to the structural complexity of real-world networks. While the theoretical foundations of many measures are intuitive, interpreting and assessing computed centrality scores---and effectively integrating them into empirical or practical research---remain context-dependent and often opaque, particularly in large and structurally heterogeneous networks.

This paper introduces Laplacian Eigenvector Centrality (LEC), a spectral-based measure designed to address these challenges by offering a versatile and theoretically grounded framework for network analysis. By leveraging the Laplacian matrix and spectral graph theory, LEC provides a tractable and flexible approach to capturing the structural roles of nodes within a network.%
\footnote{\cite{garzon2017laplacian} propose a Laplacian Eigenvector Centrality specifically for power quality analysis in electrical systems. While their approach also employs the Laplacian matrix and eigenvector-based metrics, it differs significantly in scope and methodology. They focus on the Laplacian eigenvector associated with the largest eigenvalue to assess causality and disturbance flows in power grids, whereas our framework uses multiple eigenvectors with an adjustable parameter, \textit{LEC order}. Additionally, our work provides a general framework for network analysis with applications in social and economic networks.}

The core contributions of this paper are fourfold: First, we propose a novel method to control and assess the scope of centrality analysis based on the Laplacian spectrum, establishing LEC as an analytical framework rather than merely an alternative to existing measures. Second, we investigate the statistical properties of LEC, demonstrating its scalability and robustness across networks with diverse structural characteristics. Third, we provide a microfoundation for the proposed measure, incorporating a policy-design perspective and its characterization through convex optimization. Finally, we apply LEC in an empirical setting to study the diffusion of microfinance, demonstrating how LEC complements existing measures by uncovering new insights into diffusion dynamics and network bottlenecks. Through these contributions, LEC is presented as a practical and effective tool for both theoretical and applied network analysis.

Traditional centrality measures---such as degree, closeness, and eigenvector centralities---evaluate a node's influence by aggregating information about local connectivity or specific network pathways. For instance, degree centrality counts direct connections, closeness centrality tracks shortest paths, and eigenvector centrality aggregates neighbors' centrality scores (\citealp{bloch2023centrality}; \citealp{schoch2018centrality}). While these methods provide valuable perspectives, they are often limited in their ability to account for the global structural features within a network in a transparent manner.

A distinctive feature of LEC is its adjustable \textit{LEC order}, which determines the scope of centrality measurement according to specific analytical objectives. Lower LEC orders focus on a small group of core, highly connected nodes, while higher LEC orders expand to include more peripheral nodes. This flexibility distinguishes LEC from traditional centrality measures, which typically assess node positions within a fixed scope embedded in their definitions.

Unlike traditional centrality measures, which rely on the adjacency matrix to define and aggregate nodal statistics specific to each measure, LEC uses the Laplacian matrix, which has several useful spectral properties. Technically, LEC applies dimensionality reduction by decomposing the Laplacian matrix into subspaces spanned by its eigenvectors, with eigenvalues determining the contribution of each subspace to restore the original network structure. This approach allows LEC to provide a more comprehensive assessment of node importance, capturing structural features that extend beyond simple aggregation metrics.

The mathematical foundation of LEC also offers practical guidelines for selecting the LEC order. Drawing on insights from Laplacian eigendecomposition, we propose practices using eigenvalue decay rates and cumulative eigenvalue thresholds---analogous to principal component analysis---to balance structural focus and coverage. The LEC order determines the extent to which the network's structural composition is captured, with the cumulative eigenvalue quantifying the coverage relative to the original network. This ensures that LEC-based centrality analysis remains both transparent and adaptable to diverse applications.

In Section \ref{sec:define_lec}, we establish the theoretical basis of Laplacian Eigenvector Centrality (LEC), providing key insights into its scoring interpretation, mathematical properties, and practical implementation. Formally, LEC assigns centrality scores based on the sum of squared components of the Laplacian eigenvectors. This formulation enables LEC to distribute scores according to the structural importance of nodes, as determined by the selected LEC order.

We demonstrate several fundamental properties of LEC. For example, LEC assigns identical scores to nodes with identical neighborhoods (the \textit{symmetry property}) and assigns lower scores to peripheral nodes with only one neighbor than to their direct neighbors (the \textit{periphery property}). Importantly, unlike classical measures, LEC does not necessarily preserve the neighborhood-inclusion preorder proposed by \cite{schoch2016re}, illustrating its distinct mechanisms for evaluating node positions.

In Section \ref{sec:statistical}, we examine the statistical properties of Laplacian Eigenvector Centrality (LEC), focusing on its scalability and robustness across networks with diverse structural characteristics. First, we introduce proportional LEC (pLEC), a specific implementation of LEC where the order is set proportionally to the network size. pLEC exhibits scale-invariant properties, making it effective for comparing networks of varying sizes. Second, we show that while network density influences Laplacian spectra, it does not significantly affect the distribution of pLEC scores. These observations provide practical guidelines for selecting the LEC order in empirical applications. Finally, we compare pLEC with classical centrality measures, such as eigenvector and Katz-Bonacich centralities.

Section \ref{sec:economic} develps a model where agents embedded in a network balance individual adaptation with coordination incentives, extending economic frameworks of adaptive organization to network settings. The model links agents' equilibrium actions to the network's Laplacian structure, demonstrating how network topology shapes adjustments to external shocks and responses to public information.

The analysis focuses on two quadratic optimization problems to characterize LEC in an economic context. The first examines the distribution of equilibrium responses to shocks, showing that shocks targeting central agents (as identified by higher LEC scores) induce broader but moderated adjustments, while shocks affecting peripheral agents lead to localized but intense responses. The second problem addresses optimal public information design, where a principal uses LEC to target strategically influential agents, enhancing coordination across the network. These results establish LEC as a valuable metric for quantifying agents' roles in moderating shocks and facilitating coordinated responses. Additionally, we examine the application of targeting interventions as in \cite{galeotti2020targeting}, where a generalized version of LEC, termed \textit{gLEC}, serves as an index for selecting optimal agents to stabilize the economy from shocks and minimize social loss.

Compared to \cite{ballester2006s}, which characterizes Katz-Bonacich centrality as the Nash equilibrium of a linear-quadratic network game, our approach emphasizes the structure of equilibrium adjustments rather than action levels. While Katz-Bonacich centrality captures the cumulative effects of influence along paths in the network, LEC focuses on the distribution and intensity of responses, emphasizing how network topology moderates the propagation of shocks and information. This distinction highlights LEC's unique role in identifying critical nodes for both coordination and attenuation.

Centrality measures are widely used in empirical studies to examine how network structures influence economic and social outcomes. For instance, they are applied to identify key individuals in information diffusion, study technology adoption, and evaluate peer effects in networked settings. \cite{banerjee2013diffusion} examine the role of social networks in the diffusion of microfinance across 43 villages in Karnataka, India. By introducing novel centrality measures, such as communication and diffusion centralities, they analyze how leaders' network positions facilitate the spread of information and influence non-leaders' participation in microfinance programs. Their findings highlight the importance of network centrality in understanding diffusion processes and participation patterns.

Building on this work, Section \ref{sec:microfinance} demonstrates how LEC complements existing centrality measures, offering a distinct perspective on how network structures shape diffusion outcomes. Specifically, LEC identifies structurally critical nodes whose network positions impose coordination pressures, contrasting with measures like eigenvector centrality, which emphasize influence propagation. The analysis provides new insights into the roles of leaders and network bottlenecks in determining participation rates, broadening the understanding of diffusion dynamics in networked settings.

The insights from this research have broad implications for both network theory and practical applications. By providing a framework that links the spectral properties of networks to centrality analysis, LEC offers a novel perspective on understanding network structures beyond traditional centrality measures. Its tractable approach facilitates targeted analyses of core-periphery structures, coordination pressures in social ties, and structural bottlenecks. These features make LEC particularly valuable for studying diffusion processes, designing information dissemination strategies, and evaluating resilience in social and economic networks.

\section{Related literature}\label{sec:literatrue}

Recent advancements in centrality analysis have focused on providing unified frameworks for understanding diverse measures of node importance. \cite{schoch2016re} demonstrate that many centrality indices, despite their conceptual differences, share a foundational property based on neighborhood inclusion, offering a unifying path algebra framework for centrality measures.%
\footnote{\cite{schoch2018centrality} extended this work by proposing rank probabilities and partial orders, bypassing the reliance on numerical indices to assess centrality.}
\cite{borgatti2006graph} and \cite{bloch2023centrality} further contributed by identifying how existing centrality measures aggregate local connectivity and path/walk length information in distinct ways, leading to divergences and correlations in centrality scores. Laplacian Eigenvector Centrality (LEC), by decomposing the structural information of networks through the Laplacian spectrum, diverges from classical centrality measures in significant respects. Notably, LEC does not satisfy the preservation property of the neighborhood inclusion preorder proposed by \cite{schoch2016re}, distinguishing it from traditional indices.%
\footnote{For further discussion, see Subsection \ref{sec:basic_properties}.}
Building on the framework proposed by \cite{bloch2023centrality}, which interprets centrality measures as aggregating nodal statistics, LEC can be viewed as the aggregation of squared Laplacian eigenvector components as nodal statistics, satisfying their main axioms. However, these statistics, derived from spectral decomposition, capture structural aspects distinct from path- or walk-based statistics.

Recent theoretical studies have advanced the axiomatic approach to centrality in undirected graphs, formalizing and examining the properties that centrality measures should or would satisfy.%
\footnote{This line of research has been intensively developed for directed graphs, demonstrating various consistency and monotonicity properties of centrality measures (\citealp{boldi2014axioms}; \citealp{palacios2004measurement}; \citealp{kitti2016axioms}; among others).}
For instance, \citet{boldi2023monotonicity} investigate (rank and score) monotonicity under edge additions, showing that many classical measures fail to satisfy these intuitive axioms. Other works address different aspects, such as distance-based measures (\citealp{skibski2018axioms}; \citealp{chebotarev2023selection}), influence-based measures (\citealp{dequiedt2017local}; \citealp{chebotarev2023selection}), and cooperative-game-based measures (\citealp{skibski2018axiomatic}). While this paper provides some intuitive and consistent properties for LEC, these axiomatizations also highlight how spectral-based methods deviate from path- and distance-based methodologies by violating certain axioms, offering a complementary perspective.

Correlation studies represent a traditional approach to exploring the relationships among centrality measures, relying on both real-world data and simulated networks generated from random graph models. These studies generally report high correlations among centrality measures, although exceptions arise in certain network types (\citealp{bolland1988sorting}; \citealp{faust1997centrality}; \citealp{valente2008correlated}). 
Recent works, including \citet{schoch2017correlations} and \citet{oldham2019consistency}, investigate the mechanisms by which centrality indices produce similar or distinct scores, depending on network topologies.

Examining Laplacian Eigenvector Centrality (LEC) reveals that it shows a consistently higher rank correlation with degree centrality, particularly in large-scale networks, compared to other centrality measures.%
\footnote{The high rank correlation between LEC and degree centrality does not imply that LEC acts as a direct proxy for degree centrality. Instead, LEC distorts degrees through spectral decomposition and LEC orders.}
This behavior can be attributed to the \textit{localization of Laplacian eigenvectors}, a phenomenon where the components of a Laplacian eigenvector have disproportionately large values concentrated on a subset of nodes, while the values for other nodes remain small or near zero. This concentration is not random but systematically influenced by the properties of the network, particularly the degree distribution (\citealp{mcgraw2008laplacian}; \citealp{hata2017localization}).
Related to this, eigenvector centrality derived from the adjacency matrix is also known to exhibit localization, particularly in networks with high-degree hubs or power-law degree distributions. This localization leads to a loss of meaningful distinctions among nodes, reducing its effectiveness as a centrality measure in such settings (\citealp{martin2014localization}). LEC, by relying on the (multiple) Laplacian eigenvectors, avoids these pitfalls and remains stable even in networks with pronounced hubs or clustered structures, making it a robust alternative to adjacency-based eigenvector centrality.

The Laplacian spectrum approach has significantly contributed to advancements in the fields of computer science, machine learning, and network analysis. This is because the Laplacian spectrum effectively characterizes network structure and exhibits remarkable properties under link/node removal (\citealp{milanese2010approximating}), ensemble averageability in random scale-free networks (\citealp{kim2007ensemble}), spectral clustering (\citealp{ng2001spectral}; \citealp{von2007tutorial}), and localization, as discussed above.  

In this line, \citet{qi2012laplacian} introduced Laplacian centrality, which evaluates a node's impact on the total Laplacian energy---the sum of the squared eigenvalues of the Laplacian matrix---by quantifying the energy reduction when the node and its edges are removed. In contrast, LEC uses the Laplacian eigenvectors to define indices, while the Laplacian eigenvalues are used to evaluate the scope captured through LEC orders. In this respect, this paper is closely related to Laplacian eigenmaps for dimensionality reduction by \citet{belkin2003laplacian}, which preserve local neighborhood information by projecting high-dimensional data onto a lower-dimensional manifold while maintaining the intrinsic geometric structure. While LEC can be viewed as an application of dimensionality reduction, it differs by focusing on node centrality and using only the diagonal entries of the dimension-reduced Laplacian matrix.

Economic studies have established a connection between Nash equilibria in network games and centrality measures, particularly Katz-Bonacich centrality, as a reflection of equilibrium behavior (\citealp{ballester2006s}; \citealp{calvo2009peer}).%
\footnote{In these works, the network represents payoff externalities between players. Another strand of research examines communication networks, exploring how information flows within networks influence equilibrium behaviors (\citealp{calvo2009information}) or emerge as equilibrium outcomes in strategic information transmission games (\citealp{calvo2015communication}).}
Building on this foundation, studies on monopoly pricing and consumption externalities demonstrate that optimal discriminatory pricing or induced demand profiles are characterized by Katz-Bonacich centrality (\citealp{bloch2013pricing}; \citealp{candogan2012optimal}). These findings highlight the role of centrality measures in optimizing outcomes within economic models of networked environments.

Spectral properties of networks have also been employed as a foundation for designing targeted interventions. \citet{galeotti2020targeting} propose eigenvector-based strategies to optimize interventions, emphasizing the importance of principal components of the adjacency matrix in making targeting policies effective. Similarly, \citet{tamura2025spectral} develop a Bayesian persuasion framework for adaptive organizational networks, utilizing the Laplacian spectrum to determine the dimensions of public signals. Together with these studies, this paper demonstrates how spectral graph theory extends the applicability of centrality measures, including LEC and its generalized version $\textit{gLEC}$, to intervention design and information dissemination strategies.

There is a growing literature on diffusion processes from both theoretical and empirical perspectives. On the theoretical side, \citet{akbarpour2023just} and \citet{awaya2024spreading} analyze the impact of seeding strategies on diffusion using formal models. \citet{kempe2003maximizing} propose algorithms for identifying influential nodes (``seeds'') in networks and evaluate their performance based on centrality measures. Our work contributes to this seeding problem by highlighting not only the role of network connections but also the incentive costs of adaptation and coordination incurred by agents on networks.

On the empirical side, studies such as \citet{banerjee2013diffusion}, \citet{banerjee2019using}, and \citet{beaman2021can} demonstrate the practical application of centrality measures in field experiments, emphasizing the significant influence of central nodes on information diffusion in real-world networks. While traditional centrality measures generally capture simple and positive effects of central nodes, LEC captures more complex contributions of central nodes, complementing traditional centrality-based approaches.

\section{A non-technical introduction to LEC}\label{sec:nontechnical}

Laplacian Eigenvector Centrality (LEC) provides a flexible method for analyzing centrality in networks, making it particularly valuable for systematically capturing layers of influence and connectivity. Unlike traditional centrality measures that focus on specific aspects of a node's position within a network, LEC allows researchers to adjust their focus. This flexibility enables analyses that range from identifying core influential agents to quantifying direct and indirect connectivity across both central, highly connected nodes and more peripheral ones.

A key feature of LEC is its adjustable parameter, known as the LEC order. With each increase in LEC order, the measure broadens its focus by adding an additional cumulative score of 1, which is distributed among nodes according to their connectivity. At order 0, for instance, the total LEC score is 1, evenly distributed among all nodes so that each node has a score of $\frac{1}{n}$, independent of the network. At higher orders, additional scores are allocated based on connectivity, allowing more prominent nodes to accumulate higher values. This progressive increase in scores allows researchers to control the scope of influence, moving from central nodes to more peripheral nodes as needed. 

\begin{figure}\centering
\subfloat[Order 0]{\includegraphics[height=0.33\textwidth]{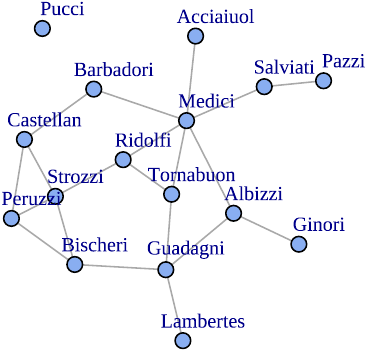}\label{fig:florentine_lec00}}\hspace{5mm}
\subfloat[Order 1]{\includegraphics[height=0.33\textwidth]{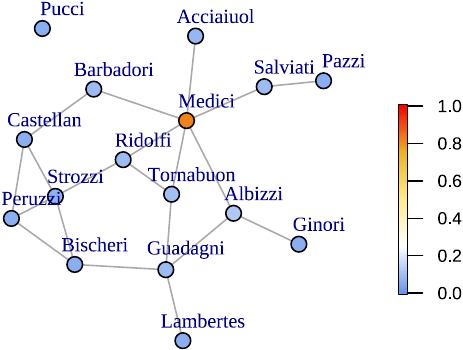}\label{fig:florentine_lec01}}\\[5mm]
\subfloat[Order 3]{\includegraphics[height=0.33\textwidth]{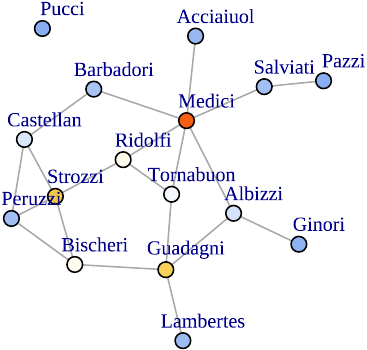}\label{fig:florentine_lec03}}\hspace{5mm}
\subfloat[Order 6]{\includegraphics[height=0.33\textwidth]{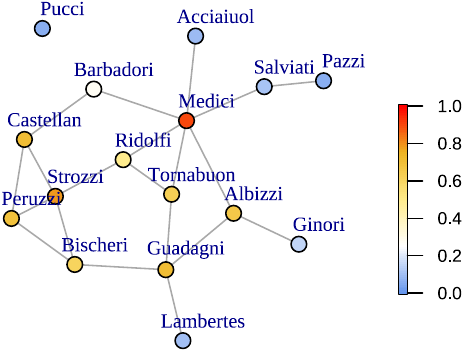}\label{fig:florentine_lec06}}
\caption{Expansion of the scope of LEC in Florentine network.}\label{fig:florentine_lec}
\end{figure}

 To illustrate the impact of increasing LEC orders, we apply this measure to the Florentine marriage network, a historical social network of influential families. As shown in Figure \ref{fig:florentine_lec}, each increase in LEC order progressively includes additional families as central within the network, capturing an increasingly comprehensive view of influence. Starting with Order 0, which assigns an equal score of $\frac{1}{16}$ to each family, LEC at order 1 distinctly identifies the Medici family as the most central, with a score of 0.85. At order 3, additional families such as the Strozzi and Guadagni emerge with scores between 0.62 and 0.66, forming a secondary layer of influence. By order 6, the LEC measure encompasses a broader set of families central to the network's core, while peripheral nodes retain lower scores. This progression demonstrates LEC's adaptability, enabling analyses that focus either on the network's core or on a wider set of influential nodes.
 
\begin{figure}
\hspace{30mm}\includegraphics[height=0.4\linewidth]{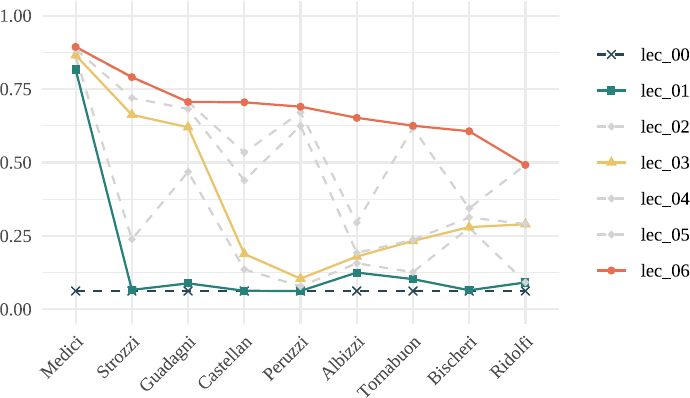}
\caption{LEC scores for different orders (Top 7 families).}\label{fig:florentine_expansion}
\end{figure}

Figures \ref{fig:florentine_expansion} and \ref{fig:florentine_other} illustrate the distinctive features of Laplacian Eigenvector Centrality (LEC) by comparing it with traditional centrality measures. Figure \ref{fig:florentine_lec} demonstrates how increasing the LEC order progressively expands the scope of centrality within the network, transitioning from a focus on core nodes to a broader range of influential agents. This progression highlights that LEC order is not merely a simple scaling parameter; rather, it reveals unique layers of centrality by capturing structural nuances across different levels of the network.

In Figure \ref{fig:florentine_other}, we present traditional centrality measures---including degree, betweenness, closeness, and eigenvector-based centralities---with scores normalized such that the Medici family's score is set to 1 for easier comparison. Except for betweenness centrality, these traditional measures typically exhibit consistent rankings and scores across prominent families. Eigenvector centrality and related Bonacich centralities, for instance, assign relatively higher scores to the Tornabuoni and Ridolfi families, both of whom are closely connected to the Medici, Strozzi, and Guadagni families (also depicted in Figure \ref{fig:florentine_lec}). This outcome aligns with the principles of eigenvector-based centralities, where a node's score is enhanced by connections to other high-scoring nodes. By contrast, LEC's analysis brings forward structural differences among these influential families, providing a more nuanced perspective that traditional centrality measures may overlook.
\begin{figure}
\hspace{30mm}\includegraphics[height=0.4\linewidth]{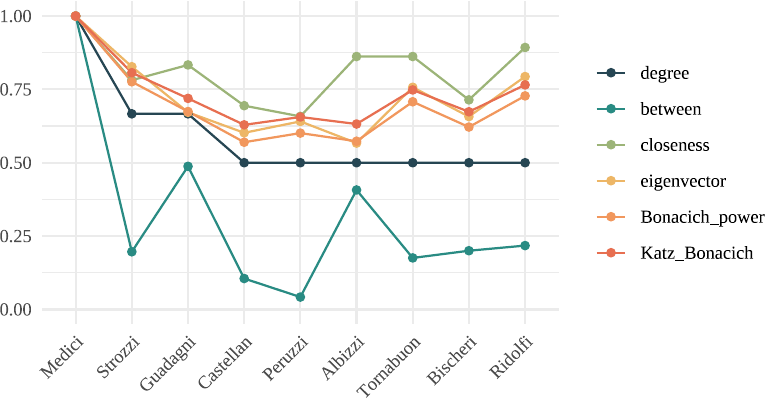}
\caption{Traditional centrality scores (Top 7 families).}\label{fig:florentine_other}
\end{figure}

In summary, this section has introduced Laplacian Eigenvector Centrality (LEC) and highlighted its adaptability in revealing different layers of network centrality through adjustable orders. The following sections will present the mathematical foundation of LEC, providing insights into its core properties and practical applications. First, we formally define LEC using Laplacian eigenvectors and examine its mathematical properties to show how it systematically captures network structure. Subsequent sections explore specific examples and the statistical properties of LEC, followed by an economic model that connects the centrality measure to the equilibrium behavior in games with strategic complementarities, and an empirical application demonstrating its relevance in real-world diffusion processes. Together, these sections expand on the concepts introduced here, offering a comprehensive view of LEC's analytical potential across diverse network settings.

\section{Laplacian eigenvector centrality: definition}\label{sec:define_lec}
\subsection{A formal introduction}\label{sec:lec_def}
Consider an undirected simple graph $G = (V, E)$, where $V = \{1, 2, \dots, n\}$ is the set of nodes (or vertices) and $E$ is the set of edges. Let $A$ denote the $n \times n$ adjacency matrix of $G$, defined such that $A_{ij} = 1$ if there is an edge between nodes $i$ and $j$, and $A_{ij} = 0$ otherwise. Since $G$ is an undirected simple graph, $A$ is symmetric and has zero diagonal entries (i.e., $A_{ii} = 0$ for all $i$).

The degree of a node $i$, denoted by $d_i$, is the number of edges connected to node $i$, calculated as $d_i = \sum_{j=1}^{n} A_{ij}$. The degree matrix of the graph $G$ is then defined as $D = \text{diag}(d_1, d_2, \dots, d_n)$, where $D$ is a diagonal matrix with the degrees of each node along its diagonal.

The Laplacian matrix $L$ of the graph $G$ is defined as $L = D - A$. The diagonal entries of $L$ are given by the degrees of the nodes, $L_{ii} = d_i$, and the off-diagonal entries are $L_{ij} = -A_{ij}$. Since $A$ is symmetric for an undirected graph, $L$ is also symmetric, and each row sum of $L$ equals zero. The Laplacian matrix $L$ is known to be positive semidefinite, meaning that all its eigenvalues, collectively referred to as the Laplacian spectrum, are nonnegative. We denote the eigenvalues in descending order as follows:
\begin{equation}
\lambda_1 \geq \lambda_2 \geq \cdots \geq \lambda_n.
\end{equation}
The smallest eigenvalue, $\lambda_n$, is always zero, and its corresponding eigenvector is proportional to the vector of ones, $\bm{q}_n \propto \bm{1}_n$. By convention, we normalize the eigenvectors $\bm{q}_1, \dots, \bm{q}_n$ so that they form an orthogonal set: each eigenvector has unit length ($\bm{q}_i \cdot \bm{q}_i = 1$ for each $i$) and they are mutually orthogonal ($\bm{q}_i \cdot \bm{q}_j = 0$ for $i \neq j$).

We introduce some notation for clarity. Let $\bm{q}_i^2$ denote the square of the vector $\bm{q}_i$.%
\footnote{Specifically, $\bm{q}_i^2$ is a vector of length $n$, where the $k$-th entry is the square of the $k$-th entry of $\bm{q}_i$, or equivalently, $\bm{q}_i^2 = \text{diag}(\bm{q}_i \bm{q}_i')$ in matrix notation.} 
For notational simplicity, we define $\bm{q}_0 = \bm{q}_n = \frac{1}{\sqrt{n}} \bm{1}_n$, so that $\bm{q}_0^2 = \frac{1}{n} \bm{1}_n$.

We now formally define the Laplacian Eigenvector Centrality (LEC) using the Laplacian eigenvectors.%
\footnote{In (\ref{eq:lec_modified}) below, we present a modified version of this definition to account for cases of eigenvalue multiplicity. Here, we assume $\lambda_{r} < \lambda_{r+1}$ for simplicity.}
\begin{definition}[LEC of order $r$]
	The Laplacian Eigenvector Centrality (LEC) of order $r$ is defined as
	\begin{equation}\label{eq:lec_def}
		\bm{c}^{LEC(r)} =  \bm{q}_0^2 + \bm{q}_1^2 + \cdots + \bm{q}_r^2.
	\end{equation}
	The LEC score for node $i$ is then expressed as 
	\begin{equation}
		c^{LEC(r)}_i = \sum_{s=0}^{r}\left[q_s(i)\right]^2
	\end{equation}
	where $q_s(i)$ is the $i$-th entry of vector $\bm{q}_s$.
\end{definition}

At order zero, each node receives an equal score of $c_i^{LEC(0)}=\frac{1}{n}$. For higher orders, LEC assigns nonnegative values up to $r$ times, determined by the squared Laplacian eigenvectors. This allocation reflects the network's structure, as nodes with greater centrality tend to accumulate higher scores at smaller LEC orders, as demonstrated in Section \ref{sec:nontechnical}.

Unlike traditional centrality measures, the interpretation of LEC scores is not straightforward due to its reliance on the spectral properties of the Laplacian matrix. This complexity arises from two primary factors. First, the eigenvectors of the Laplacian matrix reflect global structural information about the entire network, making their direct interpretation less intuitive. Second, the cumulative formulation of LEC, as the sum of squared components of multiple eigenvectors, further complicates a simple explanation of the scores. Nevertheless, these features are integral to LEC's ability to capture hierarchical layers within the network, as demonstrated in Section \ref{sec:nontechnical}. In later sections, we present how the formulation is based on dimensionality reduction of Laplacian matrices (Subsection \ref{sec:dimensionality_reduction}) and explain why it captures core components of networks by utilizing the localization effects of Laplacian eigenvectors (Subsection \ref{sec:localization}).

A technical, but important, note on the multiplicity of Laplacian eigenvalues is necessary here. Suppose $K$ is the set of indices such that $\lambda_k = \lambda_{k'}$ for any $k, k' \in K$. Let $\underline{k}$ and $\bar{k}$ denote the smallest and largest elements of $K$, respectively. Then, for $r$ such that $\underline{k} \leq r < \bar{k}$, the LEC is modified as follows:
\begin{equation}\label{eq:lec_modified}
\begin{aligned}
	\bm{c}^{LEC(r)} =&  \sum_{s = 0}^{\underline{k}-1} \bm{q}_s^2 + \frac{r-(\underline{k}-1)}{\bar{k}-(\underline{k}-1)} 
	\sum_{s = \underline{k}}^{\bar{k}} \bm{q}_s^2
\\
=& \bm{c}^{LEC(\underline{k}-1)} + \frac{r-(\underline{k}-1)}{\bar{k}-(\underline{k}-1)} 
	\left( \bm{c}^{LEC(\bar{k})}-\bm{c}^{LEC(\underline{k}-1)} \right).
\end{aligned}
\end{equation} 
This modification ensures that the LEC definition remains consistent and independent of the specific choice of eigenvectors $\bm{q}_k$ for $k \in K$ when eigenvalues have multiplicity. Throughout the paper, we simply refer the LEC as the modified version of LEC taking into account the Laplacian eigenvalue multiplicity.

For illustration, we consider a star graph of size 4, with one central node connected to three peripheral nodes. Specifically, suppose that node 1 is connected to all other nodes, while each node $i \neq 1$ is connected only to node 1. The Laplacian matrix for this graph is given by
\begin{equation}
	L = \begin{bmatrix}
		3 & -1 & -1 & -1 \\
		-1 & 1 & 0 & 0 \\
		-1 & 0 & 1 & 0 \\
		-1 & 0 & 0 & 1
	\end{bmatrix}.
\end{equation}

The eigenvalues and corresponding eigenvectors of the Laplacian matrix $L$ are presented in Table \ref{tab:star_eigen}.%
\footnote{Note that we can choose alternative orthogonal vectors in $\mathbb{R}^4$ as eigenvectors $\bm{q}_2$ and $\bm{q}_3$, which may differ from those presented in Table \ref{tab:star_eigen}. However, as shown below, the choice of eigenvectors for repeated eigenvalues does not affect the LEC scores.}
Using these eigenvectors, the LEC of order 1 is computed as
$\bm{c}^{LEC(1)} = \left[1, \frac{1}{3}, \frac{1}{3}, \frac{1}{3}\right]'$, 
where node 1 receives the maximum score of 1, and the remaining nodes equally share the remaining score of 1.%
\footnote{Note that as presented in Property \ref{prop:summation} below, at order 1, the sum of the LEC scores across all nodes equals 2.}

For orders 2 and 3, we apply the modified LEC because $\lambda_2 = \lambda_3 = 1$. Thus,
\begin{equation}\nonumber
\begin{aligned}
\bm{c}^{LEC(2)} =& \bm{c}^{LEC(1)} + \frac{1}{2} q_2^2 + \frac{1}{2} q_3^2
\\
=& 	\left[1, \frac{2}{3}, \frac{2}{3}, \frac{2}{3}\right]'.
\end{aligned}
\end{equation}
Since node 1 has already the maximum score at order 1, the LEC of order 2 equally assigns the score to the periphery nodes. At order 3, the LEC is fully expanded and results in $\bm{c}^{LEC(3)}=\bm{1}_4$.
The star graph example presented here implicitly presents several important properties of LEC, which are summarized in the next subsection. 

\begin{table}
\centering
\begin{tabular}[t]{c|c|c|c|c}
\toprule
Eigenvalues ($\lambda_i$) & $3$ & $1$ & $1$ & $0$ \\
\midrule
Eigenvectors ($\bm{q}_i$) &
$\frac{1}{\sqrt{12}} \begin{bmatrix} 3 \\ -1 \\ -1 \\ -1 \end{bmatrix}$ &
$\frac{1}{\sqrt{2}} \begin{bmatrix} 0 \\ 0 \\ -1 \\ 1 \end{bmatrix}$ &
$\frac{1}{\sqrt{6}} \begin{bmatrix} 0 \\ 2 \\ -1 \\ -1 \end{bmatrix}$ &
$\frac{1}{2} \begin{bmatrix} 1 \\ 1 \\ 1 \\ 1 \end{bmatrix}$ \\
\bottomrule
\end{tabular}
\caption{Laplacian eigenvalues and eigenvectors for a star graph of size 4.}
\label{tab:star_eigen}
\end{table}

\subsection{Basic properties}\label{sec:basic_properties}
This subsection establishes key properties of Laplacian Eigenvector Centrality (LEC), including its boundedness, symmetry for identically positioned nodes, and its evaluation of isolated and highly connected nodes. These properties offer insights into how LEC quantifies centrality by capturing the structural characteristics of networks. By examining these mathematical foundations, we clarify how LEC consistently allocates scores across varying network configurations and how it responds to specific topological features.

\subsubsection{Scoring interpretation}  
LEC scores are derived as cumulative sums of squared Laplacian eigenvector components. This structure, combined with the fundamental properties of eigenvectors, leads to several straightforward yet informative results. First, LEC scores are nondecreasing with respect to the order, ensuring that higher orders incorporate more structural information:
\begin{property}\label{prop:nondegreasing}
The LEC score is nondecreasing in its order. For any order \( r \), \(\bm{c}^{LEC(r)} \leq \bm{c}^{LEC(r+1)}\).
\end{property}

Second, LEC scores are bounded, with all nodes receiving identical scores at the extreme orders:
\begin{property}\label{prop:bound}
At order \( 0 \), all nodes receive the same score of \( \frac{1}{n} \). At order \( n-1 \), all nodes receive a score of \( 1 \):
\begin{equation}\nonumber
\bm{c}^{LEC(0)} = \frac{1}{n} \bm{1}_n \quad \text{and} \quad \bm{c}^{LEC(n-1)} = \bm{1}_n.
\end{equation}
\end{property}

Third, at any given order $r$, the sum of all LEC scores equals \( 1 + r \), indicating that the total score allocated to the network is determined by the LEC order, independent of the network's spectral properties:  
\begin{property}\label{prop:summation}  
At order \( r \), the sum of the LEC scores across all nodes equals \( 1 + r \):  
\begin{equation}\nonumber  
\sum_{j=1}^n c^{LEC(r)}_j = 1 + r.  
\end{equation}  
\end{property}

These properties provide an intuitive framework for interpreting LEC as a process of distributing a total score among \( n \) nodes based on the network structure. At order 0, all nodes receive a uniform baseline score of \( \frac{1}{n} \), reflecting the absence of structural differentiation. As the order increases, the total score grows linearly with \( 1 + r \), eventually reaching \( n \) at order \( n-1 \), where all nodes receive a score of 1. In each increase in the LEC order, one unit of ``budget'' is distributed reflecting the network's spectral components. The summation property is particularly useful when the LEC order is set proportional to the network size, as it ensures that the mean score remains constant across networks of varying scales (discussed further in Section \ref{sec:statistical}).  

\subsubsection{Symmetry and positional properties}  
The next two properties ensure that LEC consistently captures structural positions within the network.  
The symmetry property guarantees that nodes with identical structural roles within the network receive the same LEC scores.  

\begin{property}\label{prop:symmetry}
Let \(\mathcal{N}(i)\) and \(\mathcal{N}(j)\) denote the sets of neighbors of nodes \(i\) and \(j\), respectively. If \(\mathcal{N}(i) \setminus \{j\} = \mathcal{N}(j) \setminus \{i\}\), then \(c^{LEC(r)}_i = c^{LEC(r)}_j\) for any order \(r\).  
\end{property}  

The next property highlights how LEC respects hierarchical relationships between degree-1 nodes and their neighbors. Nodes with only one connection receive no more centrality than their directly connected neighbors.  

\begin{property}\label{prop:terminal}
Suppose that node $i$ has degree 1 and is connected to node $j$ (i.e., $A_{ij} = 1$). Then, for any LEC order $r$, the LEC score of node $i$ satisfies $c^{LEC(r)}_i \leq c^{LEC(r)}_j$.  
\end{property}

These symmetry and positional properties are shared by classical centrality measures such as betweenness, closeness, eigenvector, and Katz-Bonacich centralities. These measures further satisfy a monotonicity property based on neighborhood inclusion, introduced by \cite{schoch2016re} and \cite{schoch2018centrality}. 
This property states that if the neighbor set of node $i$ is a subset of the closed neighbor set of node $j$, then node $j$ should receive a centrality score at least as high as node $i$:  
\begin{equation}\nonumber
\mathcal{N}(i) \subseteq \left\{\mathcal{N}(j) \cup \{j\}\right\} \implies c_i \leq c_j.
\end{equation}  
This principle aligns with the intuition that a node with a larger or more inclusive neighborhood should, in general, be considered at least as central as a node with a smaller or less significant neighborhood.

\begin{figure}\centering
	\includegraphics[width=0.3\linewidth]{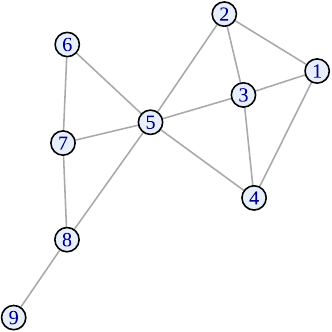}
	\caption{Violation of neighborhood-inclusion principle.}\label{fig:example_graph}
\end{figure}

To illustrate where LEC deviates from this principle, consider the graph in Figure~\ref{fig:example_graph}. Here, node \(6\) has a neighbor set \(\mathcal{N}(6) = \{5, 7\}\), while node \(7\) has a neighbor set \(\mathcal{N}(7) = \{5, 6, 8\}\). Clearly, \(\mathcal{N}(6)\) is a subset of the closed neighborhood of node \(7\), i.e., \(\mathcal{N}(6) \subseteq \{\mathcal{N}(7) \cup \{7\}\}\). Under the neighborhood-inclusion principle, node \(7\) would be expected to receive a score at least as high as node \(6\). However, the LEC scores at orders \(1\) and \(2\) are as follows:  
\[
c^{LEC(1)}_6 = 0.1322 ~>~ c^{LEC(1)}_7 = 0.1286; \quad c^{LEC(2)}_6 = 0.1322 ~>~ c^{LEC(2)}_7 = 0.1286.
\]
Contrary to the neighborhood-inclusion principle, node \(6\) receives a slightly higher score than node \(7\) at these low orders.

From a practical perspective, however, such deviations are negligible. At lower LEC orders, the majority of the total scores are concentrated on structurally prominent nodes, such as node \(3\) or node \(5\), which receive scores greater than \(0.8\). The small differences between nodes \(6\) and \(7\) are minor and have no substantial impact on the overall interpretation of centrality. 

While LEC satisfies the symmetry property shared by classical centrality measures, it does not strictly conform to the neighborhood-inclusion principle. This distinction arises because LEC relies on the spectral decomposition of the Laplacian matrix, which captures global structural patterns rather than aggregating local neighbor set relationships. However, such deviations are rare and typically involve extremely low or high LEC orders, ensuring they have negligible practical impact in network analysis.%
\footnote{Indeed, if we apply the 50\% rule to select the LEC order as introduced in Subsection \ref{sec:select_order}, the LEC score for node 7 becomes much higher than that for node 6.}

\subsubsection{Reflecting network positions}

The following properties examine how LEC evaluates nodes with specific positions or structural roles within a network. Property \ref{prop:isolated} considers isolated nodes, while Property \ref{prop:core} explores networks with fully connected hub nodes.

\begin{property}\label{prop:isolated}  
Suppose there are $k$ isolated nodes (i.e., nodes with zero degree). Then, the LEC scores for such isolated nodes are $\frac{1}{n}$ for any LEC order $r \leq n - k - 1$.  
\end{property}  

\begin{property}\label{prop:core}  
Consider a network consisting of $k$ fully connected hub nodes, each having degree $n-1$, and $n-k$ peripheral nodes with degree $k$, where each peripheral node is connected only to the hub nodes. At LEC order $k$, the hub nodes receive a score of $1$.  
\end{property} 

Property \ref{prop:isolated} highlights that isolated nodes receive a uniform baseline score of \( \frac{1}{n} \) up to a specific order. This result aligns with the interpretation of LEC as distributing scores proportionally to structural prominence; since isolated nodes have no connections, their scores remain minimal for lower orders.

Property \ref{prop:core} demonstrates how LEC identifies highly connected nodes (hubs) in networks with core-periphery structures. Specifically, when the order matches the number of hub nodes, these hubs achieve the maximum centrality score of \( 1 \). This result generalizes the star graph example discussed in Subsection \ref{sec:lec_def}, where the central node in a star achieves the upper bound score. 

\paragraph{}  
The properties outlined in this subsection establish a rigorous foundation for interpreting Laplacian Eigenvector Centrality (LEC). The scoring properties demonstrate how LEC systematically distributes scores across nodes based on their positions within the network, ensuring boundedness and consistency. The next subsection explores another aspect of LEC through eigendecomposition, highlighting the mechanisms by which LEC captures the underlying network structure.

\subsection{LEC as dimensionality reduction}\label{sec:dimensionality_reduction}
In this subsection, we provide a mathematical foundation for the Laplacian Eigenvector Centrality (LEC), which helps in understanding how it works to identify the central nodes and selecting the appropriate LEC order detailed in the next subsection. We begin with the eigendecomposition of the Laplacian matrix $L$:
\begin{equation}\label{eq:decompose_L}
L = Q \Lambda Q' = \lambda_1 \bm{q}_1 \bm{q}_1' + \lambda_2 \bm{q}_2 \bm{q}_2' + \cdots + \lambda_{n-1} \bm{q}_{n-1} \bm{q}_{n-1}' + \lambda_{n} \bm{q}_{n} \bm{q}_{n}',
\end{equation}
where $Q = \left[\bm{q}_1, \bm{q}_2, \dots, \bm{q}_n\right]$ is the $n \times n$ matrix containing the Laplacian eigenvectors as columns, and $\Lambda = \text{diag}(\lambda_1, \dots, \lambda_{n})$ is the $n \times n$ diagonal matrix with the Laplacian eigenvalues on its diagonal. The eigendecomposition in (\ref{eq:decompose_L}) shows that the Laplacian matrix $L$ can be represented as a weighted sum of $n \times n$ rank-1 matrices, with weights given by the Laplacian spectrum.%
\footnote{Each $\bm{q}_i \bm{q}_i'$ is an $n \times n$ matrix of rank 1, and the column spaces of these matrices are mutually orthogonal.}
In this formulation, $L$ is decomposed into orthogonal subspaces of rank 1, where $\lambda_i$ represents the contribution of each component $\bm{q}_i \bm{q}_i'$ to the overall composition of $L$. 

The LEC defined in (\ref{eq:lec_def}) can be expressed as $\bm{c}^{LE(r)} = \text{diag}(S)$, where
\begin{equation}
S = \bm{q}_0 \bm{q}_0' + \bm{q}_1 \bm{q}_1' + \cdots + \bm{q}_r \bm{q}_r',
\end{equation}
which can be interpreted as a weighted sum of $\{\bm{q}_i \bm{q}_i'\}_{i=0}^{n-1}$ with weights equal to 1 up to order $r$ and 0 thereafter.%
\footnote{$\bm{q}_0 \bm{q}_0'$ is a matrix with all entries equal to $\frac{1}{n}$, contributing a constant component to the LEC. We include this dimension at order 0 as a normalization convention for LEC.}
Thus, LEC is interpreted as an application of dimensionality reduction by retaining only terms with larger eigenvalues and omitting those with smaller eigenvalues.%
\footnote{For the dimensionality reduction technique, see \cite{belkin2003laplacian}.}

From such a dimensionality reduction perspective, we can define a generalized version of Laplacian Eigenvector Centrality using an arbitrary weight vector $\bm{w}$ as follows:
\begin{definition}[Generalized LEC with weight $\bm{w}$]\label{def_glec}
The generalized LEC (gLEC) with weight vector $\bm{w} = (w_0, \dots, w_{n-1})'$ is given by $\bm{c}^{gLEC(\bm{w})} = \text{diag}(S_{\bm{w}})$, where the $n \times n$ symmetric positive semidefinite matrix $S_{\bm{w}}$ is defined by
\begin{equation}
S_{\bm{w}} = w_0 (\bm{q}_0 \bm{q}_0') + w_1 (\bm{q}_1 \bm{q}_1') + \cdots + w_{n-1} (\bm{q}_{n-1} \bm{q}_{n-1}').
\end{equation}
The weight vector $\bm{w}$ is assumed to satisfy the monotonicity condition $w_0 \geq w_1 \geq \cdots \geq w_{n-1} \geq 0$.
\end{definition}

The monotonicity assumption $w_0 \geq w_1 \geq \cdots \geq w_{n-1}$ ensures that eigenvectors with higher index (corresponding lower eigenvalues in $L$) contribute progressively less to the gLEC score. While $w_0$ weights the trivial constant component $\bm{q}_0 \bm{q}_0'$, the ordering of subsequent weights aligns with the importance hierarchy in the Laplacian spectrum.

The generalized LEC can also be expressed as $\bm{c}^{gLEC(\bm{w})} = \sum_{i=0}^{n-1} w_i \bm{q}_i^2$. The standard LEC is a special case where $w_i \in \{0, 1\}$ with a monotonicity restriction, while a degree-centrality variant (differing slightly in normalization from traditional degree centrality) can be obtained by setting $w_0 = 1$ and $w_i = \lambda_i / n$, yielding
\begin{equation}
c_i^{deg} = \frac{1}{n} + \frac{d_i}{n}.
\end{equation}

An important application of gLEC arises in the economic analysis of targeted interventions, as demonstrated in Subsection \ref{sec:targeting}. In that context, the weights in gLEC are determined within the model, reflecting both the network topology captured by the Laplacian spectrum and the underlying payoff parameters. gLEC quantifies the network effects of a targeting intervention, assigning higher values to more highly connected agents who have greater influence on equilibrium outcomes. This demonstrates the utility of Laplacian spectral properties in designing policy strategies that account for network spill-over effects.

The dimensionality reduction perspective provides a unifying framework for understanding Laplacian Eigenvector Centrality (LEC) and its generalizations. By decomposing the Laplacian matrix into orthogonal components, LEC selectively incorporates spectral contributions up to a specified order, enabling the identification of structurally important nodes. The generalized LEC extends this approach by introducing flexible weighting schemes, allowing for analyses that capture different levels of network structure.

\subsection{Selecting the LEC order based on Laplacian spectrum}\label{sec:select_order}

So far, we have shown that Laplacian Eigenvector Centrality (LEC) allows analysts to control the scope of centrality analysis by adjusting the LEC order, thereby capturing different layers of influence within a network. This flexibility has been demonstrated both conceptually and mathematically. The non-technical introduction explained how the LEC order expands the scope of centrality analysis from core nodes to peripheral ones, while the formal definition revealed that the LEC order corresponds to a dimensionality reduction process on the Laplacian matrix. In this subsection, we provide a foundation and guidance on how to select and evaluate an appropriate LEC order in practice, using insights from the Laplacian spectrum.

As discussed in the last subsection, each eigenvalue $\lambda_i$ of the Laplacian matrix reflects the importance of its corresponding component in capturing the network's structure, with higher eigenvalues indicating more central, structurally significant components. Selecting an appropriate LEC order thus involves retaining only the most significant components, effectively applying dimensionality reduction to focus on primary network features while avoiding less impactful details. This approach is conceptually similar to Principal Component Analysis (PCA), where the higher eigenvalues of the covariance matrix explain more variance in the data, and lower eigenvalues are often excluded to simplify analysis and emphasize dominant patterns. By analyzing the Laplacian spectrum, practitioners can determine an LEC order that balances detail and simplicity, providing a centrality measure that aligns with their analysis goals without unnecessary complexity.

To assist practitioners in determining an appropriate LEC order, we introduce two practical criteria: observing eigenvalue decay patterns and using cumulative eigenvalue sums.

\subsubsection{Eigenvalue decay patterns}
The pattern of eigenvalue decay can provide a natural cutoff for selecting the LEC order. When the eigenvalues decrease sharply, it implies that the initial components (associated with higher eigenvalues) capture most of the network's structure, while subsequent components (with smaller eigenvalues) contribute less significantly. Therefore, a steep drop in eigenvalues suggests a point where dimensionality reduction can be applied effectively.

For instance, in the Florentine network (Figure~\ref{fig:florentine_lambda_all}), we observe a substantial gap between the first and second eigenvalues. This gap suggests that the LEC of order 1 captures a prominent structural feature of the network, supporting the conventional view from traditional centrality measures that places the Medici family in a central, influential position. In contrast, the near-equal values of the second and third eigenvalues, $\lambda_2$ and $\lambda_3$, indicate that both dimensions contribute similarly to the network's structure. Consequently, LEC order 3 might be more appropriate for identifying a second tier of central families, such as the Strozzi and Guadagni, capturing additional layers of influence in the network.

\subsubsection{Cumulative sum of eigenvalues as a threshold}
Another practical criterion for selecting the LEC order is the cumulative sum of the eigenvalues, which represents the proportion of the network structure captured by the first $r$ components. This cumulative approach helps analysts evaluate whether a chosen LEC order sufficiently covers the network's important nodes. Setting a threshold for cumulative coverage--such as 50\% or 75\%--offers a guideline for stopping at an appropriate LEC order, especially useful in larger networks to ensure the scope of analysis is neither too narrow nor too broad.

In the Florentine network (Figure~\ref{fig:florentine_cumsum}), the cumulative sum reaches approximately 50\% at order 3-4 and 75\% at order 6. These values suggest that, depending on the analysis goal, an LEC order between 3 and 6 could provide a balance between capturing essential network structure and ensuring relevant nodes are included:

\begin{description}
	\item[] Order 3-4 (50\% coverage): This range is suitable for focusing on the core structure, including the most central families.
	\item[] Order 6 (75\% coverage): This provides a broader view, capturing more peripheral layers without overly broadening the analysis scope.
\end{description}

High cumulative sums, approaching 100\%, should generally be avoided, as they reduce the variation in LEC scores among nodes, flattening the distinctions between central and peripheral nodes.

\begin{figure}\centering 
\subfloat[Laplacian eigenvalues]{\centering \includegraphics[width=0.45\textwidth]{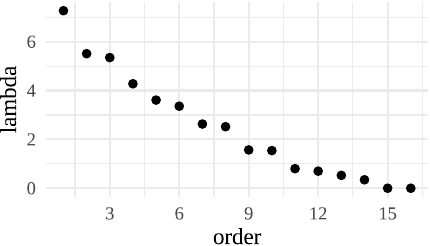}\label{fig:florentine_lambda_all}}
\quad
\subfloat[Cumulative sum of Laplacian eigenvalues]{\centering \includegraphics[width=0.45\textwidth]{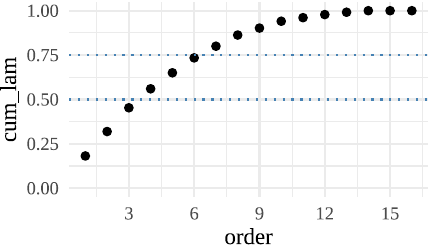}\label{fig:florentine_cumsum}}
\caption{Laplacian spectrum as a guide for the LEC order choice.}
\end{figure}

In summary, the choice of LEC order should reflect the objectives of the analysis and the network's structural characteristics. For small to moderate networks, examining eigenvalue decay patterns may be sufficient to identify an effective LEC order. In larger networks, cumulative eigenvalue thresholds provide a practical guide, with 50\% to 75\% cumulative coverage often serving as reasonable stopping points. By following these guidelines, analysts can harness the flexibility of LEC to capture centrality at a desired level of detail, aligning with the specific demands of their research or application context.

\section{Statistical properties of LEC and pLEC}\label{sec:statistical}

Understanding the statistical properties of Laplacian Eigenvector Centrality (LEC) and its proportional variant (pLEC) is essential for interpreting their behavior across networks with diverse structural characteristics. This section examines three key questions. First, how does network size affect LEC, and what motivates the use of pLEC as a scalable alternative? Second, how are LEC distributions influenced by degree-related characteristics, such as average degree and degree distribution? Third, how does LEC relate to other centrality measures, and what unique insights does it provide about network structure?
To address these questions, we employ random graph models, specifically Erd\H{o}s--R\'enyi (ER) and Barab\'asi--Albert (BA) networks.%
\footnote{The Erd\H{o}s--R\'enyi (ER) model generates random graphs where each pair of nodes is connected with probability \( p \), resulting in a Poisson degree distribution. The Barab\'asi--Albert (BA) model produces scale-free networks through preferential attachment, yielding a heavy-tailed degree distribution with a small number of highly connected hub nodes. For a formal introduction, see \cite{albert2002statistical}.}

\subsection{Network size effects and scalability}

\begin{figure}
    \centering
    \includegraphics[width=0.8\textwidth]{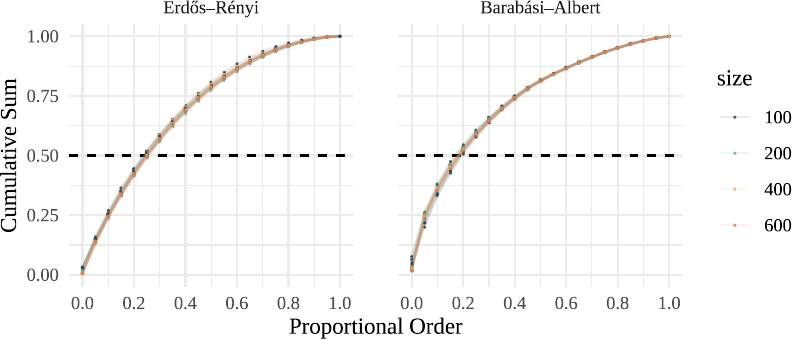}
    \caption{Cumulative sum of Laplacian eigenvalues across network sizes.}
    \label{fig:cumsum_size}
\end{figure}

Real-world networks often vary significantly in scale, requiring centrality measures that account for differences in network size. This subsection examines how scaling the LEC order proportionally to network size (\textit{pLEC}) ensures structural comparability. Using Erd\H{o}s--R\'enyi (ER) and Barab\'asi--Albert (BA) models, we demonstrate the scale-invariant properties of Laplacian eigenvalues and pLEC scores and their ability to reflect network topology.

Figure \ref{fig:cumsum_size} shows the cumulative distribution of Laplacian eigenvalues for networks of varying sizes (100 to 600 nodes), with 10 networks generated for each parameter set. The x-axis represents the eigenvalue index scaled by network size, \( x = k/n \), and the y-axis shows the cumulative sum of eigenvalues normalized to 1, \( y = \sum_{i=1}^k \lambda_i / \sum_{i=1}^n \lambda_i \).
Both Erd\H{o}s--R\'enyi (ER) and Barab\'asi--Albert (BA) networks exhibit smooth, consistent curves that are independent of network size, demonstrating the scale-invariance of the Laplacian spectrum. The cumulative sum reaches 50\% at approximately 20\% of the network size, suggesting that setting the LEC order to 0.2 times the network size captures nearly half of the network structure. Compared to ER networks, the BA networks show a slightly steeper initial curve, reflecting the influence of high-degree hubs characteristic of heavy-tailed degree distributions.

\begin{figure}
    \centering
    \includegraphics[width=0.8\textwidth]{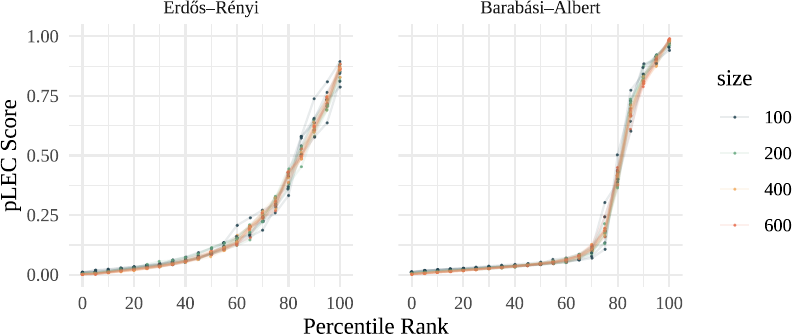} 
    \caption{Percentile distribution of pLEC scores across network sizes.}
    \label{fig:plec_size}
\end{figure}

Figure \ref{fig:plec_size} illustrates the distribution of LEC scores where the order is set to 20\% of the network size (20\%-pLEC).%
\footnote{To formally define pLEC, the \( s\%\)-pLEC is expressed as:
\begin{equation}\nonumber
\mathbf{c}^{\text{pLE}(s\%)} = \bm{q}_0^2 + \bm{q}_1^2 + \cdots + \bm{q}_{\lceil \frac{s}{100}n \rceil}^2,
\end{equation}
where \(\bm{q}_i^2\) is the component-wise square of the \( i \)-th Laplacian eigenvector, and \( s \) denotes the percentage of the network size.}
The x-axis represents the percentile rank of nodes, enabling normalized comparisons across networks. In ER networks, pLEC scores are smoothly distributed, reflecting their relatively uniform connectivity, where the degree distribution follows a Poisson distribution. In contrast, BA networks display polarized distributions: the bottom 70\% of nodes have near-zero scores, while the top 15\% attain disproportionately high scores. This pattern reflects the heavy-tailed degree distribution inherent to BA networks, where a few high-degree hubs account for a large share of the network's connectivity

By ensuring that the same proportion of the network structure is analyzed, pLEC facilitates meaningful cross-network comparisons while maintaining structural consistency. As shown in Figure \ref{fig:plec_size}, pLEC effectively accommodates differences in network size and captures the unique structural features of different network topologies. 

\subsection{Effects of average degree and degree distribution}

This subsection examines how network density influences Laplacian spectrum and pLEC scores using ER and BA models. We also compare the practical implications of applying the 50\% rule versus a proportional rule for determining LEC order, emphasizing their relevance for network analysis in varying contexts. The analysis helps to understand what aspects of network structure Laplacian spectrum and eigenvectors can capture.

\begin{figure}
    \centering
    \includegraphics[width=0.8\textwidth]{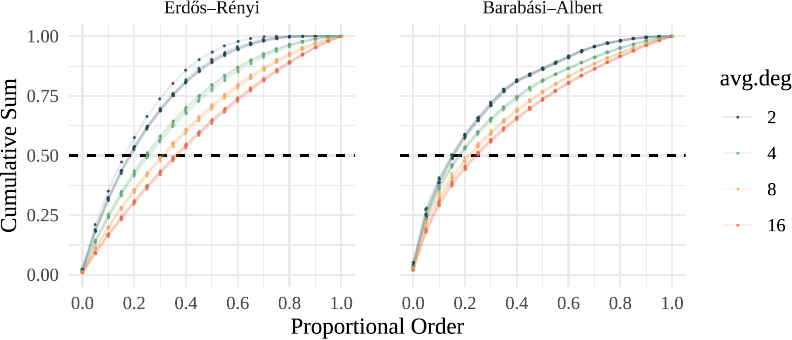} 
    \caption{Cumulative sum of Laplacian eigenvalues across densities.}
    \label{fig:cumsum_density}
\end{figure}

\subsubsection{Laplacian spectrum and network density}  
Figure \ref{fig:cumsum_density} presents the cumulative sum of Laplacian eigenvalues for networks with varying average degrees (\( \text{avg.deg} = 2, 4, 8, 16 \)). For both ER and BA models, increasing the average degree shifts the cumulative curves downward, indicating that higher densities distribute spectral contributions more evenly across a larger number of eigenvalues. This pattern reflects the denser connectivity in high-degree networks, where structural information becomes less concentrated in the leading eigenvalues. 

The impact of average degree is more pronounced in ER networks than in BA networks. In ER networks, the uniform random connection pattern spreads eigenvalue contributions more evenly as density increases. In contrast, BA networks exhibit a smaller shift, as their heavy-tailed degree distribution concentrates much of the spectral contribution in a smaller subset of dominant eigenvalues associated with high-degree hubs. This contrast highlights how network topology interacts with density to shape spectral structures.

Figure \ref{fig:plec_deg} shows the distribution of pLEC scores computed using a fixed 20\%-pLEC order. The x-axis represents the percentile rank of nodes, normalized for comparison across graphs. In ER networks, the distribution of pLEC scores remains largely invariant to changes in average degree, reflecting the consistent, homogeneous structure of these networks. In BA networks, however, pLEC scores are moderately affected by changes in density. For higher average degrees (e.g., 8 or 16), the scores become sharply polarized, with the top 15\% of nodes attaining high scores and the bottom 75\% receiving near-zero scores. By contrast, for lower average degrees (e.g., 2 or 4), the top-ranked nodes still maintain high scores, but the remaining scores decline more gradually, leading to a smoother distribution. Indeed, the increase in average degree (or density) in BA networks expands the core group of equally influential nodes, and the pLEC scores effectively capture these changes in network connectivity.
\begin{figure}
    \centering
    \includegraphics[width=0.8\textwidth]{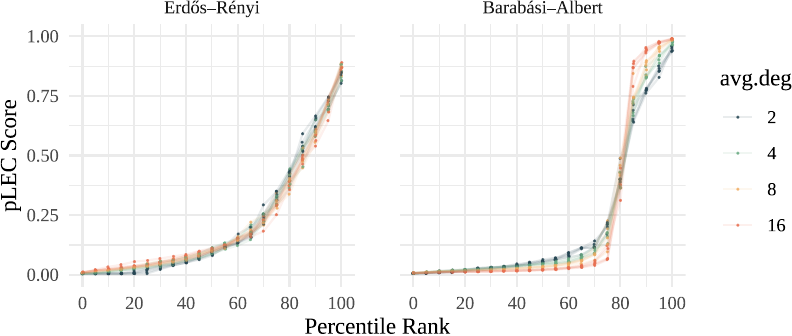}
    \caption{Distribution of pLEC scores across densities.}
    \label{fig:plec_deg}
\end{figure}

\subsubsection{Comparison of LEC order rules}  
Figure \ref{fig:plec_cumsum} illustrates the distribution of pLEC scores when the LEC order is determined using the 50\% rule, where the order is set to capture 50\% of the cumulative Laplacian eigenvalues for each network. Compared to the fixed 20\%-pLEC rule, the 50\% rule produces LEC score distributions that are more sensitive to network density in both models. This reflects the downward shift in cumulative eigenvalues, which increases the pLEC order required to capture 50\%.

The choice between these rules depends on the application context. The proportional rule (e.g., 20\%-pLEC) ensures consistency across networks of varying sizes, making it more suitable for large-scale cross-network studies. By contrast, the 50\% rule adapts to the spectral properties of individual networks, offering a more tailored but potentially less comparable analysis.

\begin{figure}
    \centering
    \includegraphics[width=0.8\textwidth]{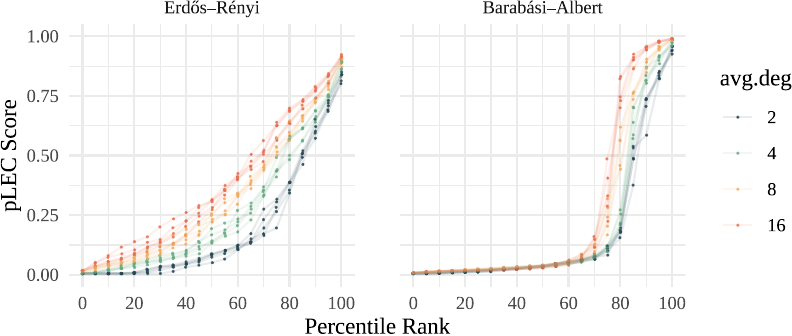}
    \caption{Distribution of pLEC scores under the 50\% rule.}
    \label{fig:plec_cumsum}
\end{figure}

\subsection{Relationship between LEC and other centrality measures}\label{sec:relation}

This subsection examines how pLEC compares to classical centrality measures, such as eigenvector and Katz-Bonacich centralities, with respect to their relationship with degree and structural differentiation. Through statistical analyses and visualizations, we highlight the unique aspects of pLEC, focusing on its correlation with degree rankings, contrasting mechanics with eigenvector centrality, and robustness to localization effects.

\subsubsection{Relationship between degree and centrality}
We analyze the relationship between centrality measures and node degrees, focusing on the extent to which centrality scores correlate with degree rankings. Figure~\ref{fig:boxplot_correlation} presents a box plot of pairwise correlations between centrality scores and node degrees across 50 networks, generated using the ER and BA models.%
\footnote{Here, we compute pLEC based on the 50\% rule demonstrated in Section \ref{sec:select_order}.} For each network type, Pearson's correlation coefficient and Kendall's $\tau$ (rank correlation) are shown.

\begin{figure}
    \centering
    \includegraphics[width=0.8\textwidth]{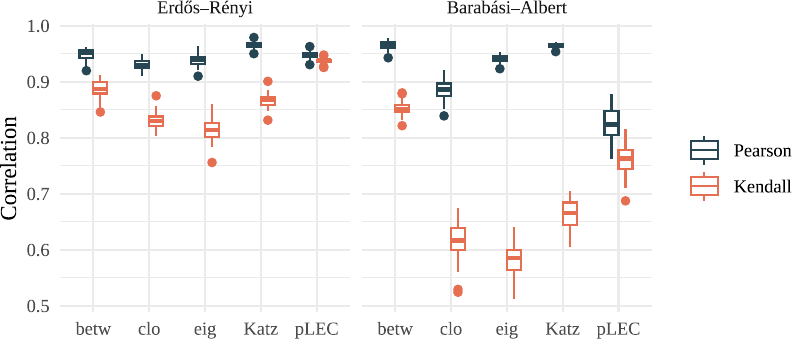}
    \caption{Box plot of pairwise correlation on 50 networks.}
    \label{fig:boxplot_correlation}
\end{figure}

The figure reveals several key patterns. First, for classical centrality measures such as closeness, eigenvector centrality, and Katz-Bonacich centrality, there is a prominent gap between Pearson and Kendall correlations, especially for BA networks. These measures exhibit high Pearson correlation coefficients, indicating a strong linear relationship with degree magnitude, but lower Kendall's $\tau$ values, reflecting weaker alignment with degree rankings. In contrast, pLEC shows relatively smaller gaps between Pearson and Kendall correlations across both ER and BA networks. Furthermore, while pLEC exhibits lower Pearson correlations compared to other measures, it achieves relatively higher Kendall's $\tau$, highlighting its stronger alignment with degree rankings.

These observations suggest that pLEC emphasizes ordinal relationships (rank consistency) over proportional scaling with degree magnitude. By contrast, other measures such as eigenvector and Katz-Bonacich centralities are more sensitive to degree magnitude but less consistent in preserving rank order. From an applied perspective, pLEC provides a distinctive approach to quantifying node centrality, combining sensitivity to degree with structural differentiation in a manner distinct from classical measures.

\subsubsection{Structural aspects captured by pLEC vs. eigenvector centrality}  
Figure~\ref{fig:scatter_plot_eig} shows the scatter plots of eigenvector centrality and pLEC scores for nodes grouped by degree in ER and BA networks of size 200. In ER networks, a negative correlation appears for nodes with intermediate degrees (e.g., degrees 7--9). This trend reflects the contrasting mechanics of eigenvector centrality and pLEC. Eigenvector centrality assigns higher scores to nodes connected to highly central neighbors, emphasizing influence propagation through densely connected hubs. In contrast, pLEC assigns higher scores to nodes connected to less central neighbors when degree and other factors are held constant. This distinction arises from pLEC's use of spectral decomposition of the Laplacian matrix, where dimensionality reduction emphasizes nodes necessary to capture the overall structure of the network.

\begin{figure}
    \centering
    \includegraphics[width=0.8\textwidth]{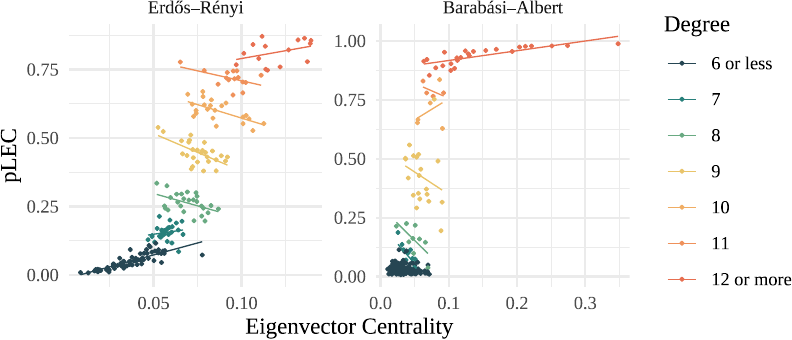}
    \caption{Plot of eigenvector centrality and pLEC.}
    \label{fig:scatter_plot_eig}
\end{figure}

In BA networks, the heavy-tailed degree distribution results in a few high-degree nodes dominating the structure. Eigenvector centrality reflects this skewness by assigning significantly higher scores to top-degree nodes, with most of the variation in scores occurring among nodes with very high degrees (e.g., degree 12 or more).%
\footnote{This is consistent with the observation in Figure \ref{fig:boxplot_correlation}, where eigenvector centrality exhibits high Pearson correlation coefficients but low rank correlation in BA networks.}
By contrast, pLEC, under the 50\% rule for determining the LEC order, adjusts its scope to capture the key structural features necessary to represent the network effectively. This results in a more gradual distribution of scores among high-degree nodes, ensuring that centrality is not overly concentrated on the most connected nodes but instead reflects the broader network hierarchy.

\subsubsection{Localization of eigenvector centrality in clustered networks}\label{sec:localization}

Localization in eigenvector centrality is a phenomenon that occurs when centrality scores disproportionately concentrate on a small subset of nodes, often overemphasizing specific regions of a network (\citealp{martin2014localization}). This issue is particularly pronounced in networks with clusters, where nodes within a single dominant cluster receive high scores, while those in other clusters, even if only slightly less connected, are assigned substantially lower scores. Such localization can obscure the contributions of nodes in less-connected or sparsely connected regions, reducing the utility of eigenvector centrality in clustered networks.

To illustrate this effect, we compare eigenvector centrality and pLEC in a synthetic clustered network, shown in Figure~\ref{fig:localization}.%
\footnote{The clustered network consists of \( k = 5 \) clusters, each generated as an Erd\H{o}s--R\'enyi (ER) graph with \( n = 50 \) nodes and connection probability \( p = 0.1 \). Clusters are rewired with a 5\% probability to introduce sparse inter-cluster edges. In the figure, node color indicates the centrality score (red = high, blue = low), and node size represents degree.}  
Figure~\ref{fig:localization}(a) demonstrates how eigenvector centrality localizes in the network, concentrating scores within one dominant cluster while nodes in other clusters receive much lower scores. This occurs because eigenvector centrality amplifies scores for nodes connected to already high-scoring neighbors, prioritizing regions with dense internal connections.

\begin{figure}\centering
	\subfloat[Eigenvector centrality]{\includegraphics[width=0.3\linewidth]{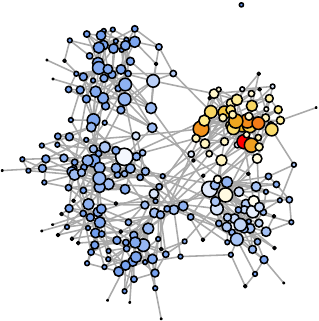}}\quad \quad \quad
	\subfloat[pLEC (50\%)]{\includegraphics[width=0.3\linewidth]{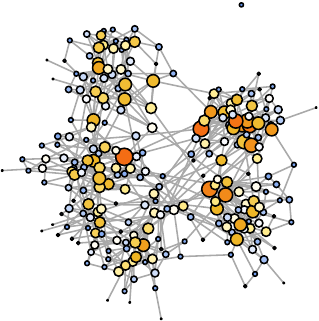}}
	
	\caption{Localization of eigenvector centrality in a clustered network.}\label{fig:localization}
\end{figure}

By contrast, Figure~\ref{fig:localization}(b) shows that pLEC mitigates this effect by capturing central nodes within each cluster.
Owing to the localization of Laplacian eigenvectors (\citealp{hata2017localization}), pLEC distributes scores more evenly across nodes with similar degrees regardless of clusters, ensuring that no single region disproportionately dominates the centrality distribution.

This comparison highlights the robustness of pLEC against localization effects in clustered networks. Unlike eigenvector centrality, which is prone to overemphasizing a single dominant region, pLEC provides a more balanced representation of node centrality across the entire network, making it particularly effective for analyzing networks with distinct structural clusters.

\section{An economic foundation: adaptation versus coordination tradeoff}\label{sec:economic}

This section develops an economic model by extending the adaptive organization framework in \cite{dessein2006adaptive} and \cite{alonso2008does} to network settings. In this model, agents balance individual adaptation preferences with coordination incentives arising from their social connections. The quadratic specification presented below provides a tractable yet flexible framework for analyzing how network structure influences equilibrium behavior in the presence of strategic complementarities.

\subsection{Model setup}\label{sec:economic_model_setup}

Consider a network with $n$ agents, indexed by $i = 1, 2, \dots, n$. Each agent $i$ chooses an action $a_i \in \mathbb{R}$ to balance adaptation to a local state $\theta_i$ with coordination incentives from their connected neighbors. Specifically, the payoff function for agent $i$ is defined as:
\begin{equation}
u_i(a_i, a_{-i}, \theta_i) = -(a_i - \theta_i)^2 - \beta \sum_{j} g_{ij} (a_i - a_j)^2,	
\end{equation}
where $\beta \geq 0$ represents the strength of coordination incentives, and $A = [g_{ij}]$ is the adjacency matrix of the network. Here, $g_{ij} = 1$ if agents $i$ and $j$ are connected, and $g_{ij} = 0$ otherwise.

This payoff function reflects two sources of losses. The first term, $-(a_i - \theta_i)^2$, penalizes differences between $a_i$ and $\theta_i$, which represents agent $i$'s preferred action when isolated from the network. The second term, $-\beta \sum_{j} g_{ij} (a_i - a_j)^2$, penalizes deviations between $a_i$ and the actions $a_j$ of connected neighbors, weighted by $\beta$. This term reflects the coordination pressure within the network, encouraging agents to align their actions with those of their neighbors.

The balance between adaptation to local states and coordination with neighbors is highlighted in each agent's best response function. The best response of agent $i$ is given by:
\begin{equation}
a_i = \frac{\theta_i + \beta \sum_{j=1}^{n} g_{ij} a_j}{1 + \beta \sum_{j=1}^{n} g_{ij}}.
\end{equation}
This expression shows that each agent's optimal action is a weighted average of their local state $\theta_i$ and the actions of their neighbors $\{ a_j : g_{ij} = 1 \}$. As the number of connections $i$ has increases, the weight on $\theta_i$ decreases, indicating that agents with more connections incur stronger coordination pressures. This phenomenon, referred to here as ``entangling social ties,'' suggests that densely connected agents are more inclined to align their actions with those of their neighbors, reducing the emphasis on their individual preferences.

The equilibrium of this network game is characterized by a linear function of the vector of local states, $\bm{\theta}$. Specifically, the equilibrium action vector $\bm{a}^*$ is given by:
\begin{equation}\label{eq:equilibrium}
\bm{a}^* = (I_n + \beta L)^{-1} \bm{\theta},	
\end{equation}
where $L$ is the Laplacian matrix associated with the adjacency matrix $A$. This formulation reveals that the equilibrium action profile depends on both the network structure (through $L$) and the local states of each agent. The Laplacian matrix captures the connectivity and positions within the coordination network, shaping how each agent adjusts their actions in response to changes in the local states of others.

\subsubsection{Comparison with Katz-Bonacich centrality}

Katz-Bonacich centrality, as developed in \cite{ballester2006s} and \cite{calvo2009peer}, captures an agent's network influence based on their equilibrium action level within a networked game. These studies consider a network game with quadratic preferences but employ a different specification:
\begin{equation}
u_i (x_i, x_{-i}) = x_i - \frac{1}{2\beta} x_i^2 + \sum_{j} g_{ij} x_i x_j,	
\end{equation}
where $x_i$ represents the action chosen by agent $i$. This payoff function combines a quadratic cost term for the agent's action level with linear interaction terms that capture peer effects among connected agents. 

Unlike the adaptation versus coordination trade-off in our model, this specification emphasizes strategic complementarity in agents' actions: higher actions by an agent's neighbors incentivize agent $i$ to increase $x_i$ as well. This strategic complementarity creates a feedback loop in which agents' actions reinforce one another. In this framework, the adjacency matrix $A$ captures both direct and indirect influence paths, with Katz-Bonacich centrality quantifying each agent's equilibrium action as a measure of their prominence within the network and the extent to which their behavior is amplified by connections.

The equilibrium action profile in this model is given by:
\begin{equation}
\bm{x}^* = (I_n - \beta A)^{-1} (\beta \bm{1}_n),	
\end{equation}
where $A$ is the adjacency matrix and $\bm{1}_n$ is a vector of ones. This formulation indicates that each agent's equilibrium action is influenced by the entire network structure, incorporating effects from both direct neighbors and indirect connections throughout the network. Katz-Bonacich centrality thus captures how an agent's position within the network determines their equilibrium behavior, reflecting a cumulative measure of influence through immediate connections and higher-order paths.

\subsection{LEC as the entangling social ties}

We develop two quadratic optimization problems to analyze how network structure influences agents' equilibrium responses to external shocks and the design of public information. These problems provide an economic foundation for Laplacian Eigenvector Centrality (LEC) by demonstrating its role in quantifying the influence of network connectivity and positions on equilibrium behavior. Specifically, the optimization problems reveal how LEC identifies central agents who play pivotal roles in distributing the impacts of shocks across the network and coordinating agents' responses through information disclosure. This analysis establishes an economic interpretation of LEC, highlighting its importance in determining how networked agents respond to external interventions and shared signals.

\subsubsection{Analyzing equilibrium responses to shocks}\label{sec:attenuating_shocks}

This first optimization problem examines how the direction of an external shock affects the range and intensity of equilibrium adjustments across a network. In particular, we investigate how shocks directed at central agents with extensive connections induce broader but more moderated responses across the network, while shocks directed at peripheral agents lead to more localized and intense adjustments. This distinction highlights the importance of network centrality in shaping the distribution of agents' equilibrium responses to external disturbances.

To formalize this analysis, we set up a quadratic optimization problem to identify the shock direction that minimizes aggregate deviations in equilibrium actions. The optimization problem is defined as:
\begin{equation}
\min_{\bm{\theta}} \bm{a}' \bm{a} \quad \text{s.t.}~ \bm{\theta} \in \mathcal{B} \equiv \{\tilde{\bm{\theta}} \in \mathbb{R}^n : \tilde{\bm{\theta}}' \tilde{\bm{\theta}} = 1 \},	
\end{equation}
where $ \bm{a} $ represents the vector of agents' equilibrium responses to a shock applied in the direction of $\bm{\theta}$, and $\mathcal{B}$ is a unit ball constraint that normalizes the shock's magnitude. By minimizing the total deviation $ \bm{a}' \bm{a} $, this problem identifies the shock direction that is most effectively distributed across the network in equilibrium, reflecting the network's capacity to dilute the shock impact.

Applying the Courant-Fischer min-max theorem via the Rayleigh quotient, we find that the solution to this optimization problem aligns with the eigenvector $ \bm{q}_1 $ associated with the largest eigenvalue $ \lambda_1 $ of the Laplacian matrix $ L $.%
\footnote{For details, see \cite{horn2013matrix}.}
This result implies that shocks aligned with $ \bm{q}_1 $ produce the smallest aggregate deviation in equilibrium actions. The components of $ \bm{q}_1 $ indicate the degree to which each agent contributes to the network's distribution of adjustments, with larger components corresponding to agents who play key roles in moderating the shock's effect.

To verify that shocks in the direction of $ \bm{q}_1 $ minimize equilibrium deviations, consider an initial equilibrium state $\bm{a}^* $ and apply a small perturbation $ \delta $ along $ \bm{q}_1 $, so that the new state vector becomes $ \bm{\theta}^{**} = \bm{\theta} + \delta \bm{q}_1 $. The corresponding equilibrium action vector $\bm{a}^{**}$ is given by:
\begin{equation}
\bm{a}^{**} = (I_n + \beta L)^{-1} \bm{\theta}^{**} = (I_n + \beta L)^{-1} \bm{\theta} + \delta (I_n + \beta L)^{-1} \bm{q}_1.
\end{equation}
Since $ L \bm{q}_1 = \lambda_1 \bm{q}_1 $, we have:
\begin{equation}
(I_n + \beta L)^{-1} \bm{q}_1 = \frac{1}{1 + \beta \lambda_1} \bm{q}_1,
\end{equation}
indicating that a shock in the $ \bm{q}_1 $ direction is attenuated by a factor of $ \frac{1}{1 + \beta \lambda_1} $, representing the greatest possible attenuation in the network.

Consequently, shocks aligned with $ \bm{q}_1 $ result in minimal equilibrium impact across the network due to the extensive connectivity among agents. In this case, agents' coordination incentives, as specified by their best response functions, distribute the shock's influence broadly yet moderately across their connections. This finding demonstrates that the $ \bm{q}_1 $ direction corresponds to the network's greatest capacity to mitigate the direct effects of disturbances.

Further insights emerge from additional eigenvectors of the Laplacian, such as $ \bm{q}_2 $, which corresponds to the second-largest eigenvalue. While $ \bm{q}_1 $ represents the direction of maximal shock absorption, $ \bm{q}_2 $ and subsequent eigenvectors highlight secondary structural components, with each eigenvector revealing distinct substructures in the network.

To interpret each agent's role in response to shocks, we examine the component-wise squares of the eigenvectors $ \bm{q}_1 $, $ \bm{q}_2 $, and others. These squared values quantify the magnitude of each agent's involvement in a shock along a particular eigenvector direction, with larger values indicating agents more directly affected by shocks in that direction. Squaring each component provides a measure of influence that disregards sign, forming the basis of Laplacian Eigenvector Centrality (LEC) as a metric for how shocks are distributed across the network.

\subsubsection{Optimal design of public information}

In the second problem, we address an information design problem where a principal aims to influence agents' decisions within a network by selectively providing public signals. The principal, aware of the network structure and the interdependencies among agents, seeks to enhance coordination among the agents, who base their actions on the specific relationships defined by the network. By strategically disclosing information, the principal can guide agents' behavior to align more closely with the principal's objectives.

The principal's objective function is specified as:
\begin{equation}
v(\bm{a}, \bm{\theta}) = -\sum_{i=1}^n (a_i - \theta_i)^2 - \tilde{\beta} \sum_{i=1}^n \sum_{j \neq i} (a_i - a_j)^2,
\end{equation}
where $a_i$ denotes the action taken by agent $i$, and the state vector $\bm{\theta}$ is assumed to follow a normal distribution, $\bm{\theta} \sim N(\bm{0}_n, I_n)$.%
\footnote{The analysis can also be extended to correlated state distributions. See \cite{tamura2025spectral}.}
The first term represents adaptation losses, capturing deviations between agents' actions and their local states. The second term represents coordination losses, where $\tilde{\beta}$ reflects the principal's valuation of coordination relative to adaptation.

We formulate the principal's problem as a standard Bayesian persuasion problem, where the principal commits to a public signal that maps the realized state to disclosed statistics, thereby determining the distribution of public posterior beliefs. Let $\hat{\bm{\theta}}$ denote the public expectation of the state, and $\bm{a}(\hat{\bm{\theta}})$ be the equilibrium action profile, given by a linear function of $\hat{\bm{\theta}}$. The principal's objective is then to maximize $\mathbb{E}\left[ v(\bm{a}(\hat{\bm{\theta}}), \hat{\bm{\theta}})\right]$.

In the case of quadratic-Gaussian specifications, it is established in the literature that the optimal signal can be expressed as a linear function of the state.%
\footnote{See \cite{tamura2018bayesian} and \cite{miyashita2023lqg}.} Specifically, the principal computes the eigenvectors of the Laplacian matrix associated with the network structure and creates statistics as linear combinations of the state, \( m_i = \bm{q}_i' \bm{\theta} \), where each \( \bm{q}_i \) is an eigenvector of the Laplacian corresponding to the $i$-th largest eigenvalue. Among the statistics \( \{m_0, m_1, \dots, m_{n-1}\} \), the principal selectively discloses only a subset. In particular, a statistic \( m_k \) is disclosed if it satisfies the criterion:
\begin{equation}
\frac{1}{2n}+\frac{\beta}{n} \lambda_k \geq \tilde{\beta}.
\end{equation}
The number of disclosed statistics decreases as \( \tilde{\beta} \) increases, meaning the principal reveals fewer statistics when placing a higher weight on adaptation relative to coordination.

If it is optimal to disclose \( r+1 \) statistics, the informativeness of the signal, measured by the variability of conditional expectations, is given by 
\begin{equation}\label{eq:informativeness}
	\mathbb{E}[\hat{\bm{\theta}}\hat{\bm{\theta}}' ] = \bar{Q}_r \left( \bar{Q}_r' \bar{Q}_r \right)^{-1} \bar{Q}_r',
\end{equation}
where \( \bar{Q}_r = \left[ \bm{q}_0, \bm{q}_1, \dots, \bm{q}_r \right] \) is the \( n \times (1+r) \) matrix of the selected eigenvectors. The diagonal entries of (\ref{eq:informativeness}) correspond to the Laplacian Eigenvector Centrality (LEC) of order \( r \), indicating that the principal's optimal signal should focus on central agents to enhance coordination effectively.

In this way, selecting signals based on LEC provides a structured approach to information design, allowing the principal to enhance coordination across the network by targeting the most strategically influential agents.

\subsection{Using gLEC for targeting policy design}\label{sec:targeting} 

This subsection provides an economic foundation for the generalized Laplacian eigenvector centrality (gLEC) defined in Definition \ref{def_glec}. Consider the setting in Subsection \ref{sec:economic_model_setup}, where agents are embedded in a social and economic network, and their equilibrium actions balance adaptation and coordination incentives as characterized by \eqref{eq:equilibrium}. Suppose the economy experiences a global shock that uniformly increases the state vector. The government mitigates the shock's effects by implementing a targeted intervention policy that neutralizes its impact on a single agent. By linking the optimal intervention target to gLEC, we demonstrate how the network's structure influences the reduction of aggregate social loss.  

To formalize the analysis, we begin by describing the initial conditions and the effects of the shock. Initially, the economy is in a steady state, represented by $\bm{\theta} = \bm{0}_n$. Under this steady state, the equilibrium action vector is $\bm{a} = \bm{0}_n$. A global shock occurs, uniformly increasing the state vector to $\bm{\theta} = \bm{1}_n$. In the absence of any intervention, this shock induces an equilibrium action vector of $\bm{a} = \bm{1}_n$. The social loss resulting from the shock is defined as the aggregate quadratic disturbance in equilibrium actions, $\bm{a}' \bm{a} = \sum_{i} a_i^2$. Without intervention, this social loss equals $n$ due to the uniformity of the shock and the resulting equilibrium actions.  

To mitigate the effects of the shock, the government can intervene by targeting a single agent $i$ from the feasible set $\mathcal{T} \subseteq \{1, 2, \dots, n\}$. This intervention neutralizes the shock's impact on agent $i$, modifying the state vector to $\widetilde{\bm{\theta}}(i) = \bm{1}_n - \bm{e}_i$. Here, $\bm{e}_i$ is a unit vector of length $n$ with a 1 in the $i$-th entry and 0 elsewhere. This adjustment reduces agent $i$'s local state to zero while leaving the state of all other agents unchanged. The government minimizes the social loss, $\bm{a}' \bm{a}$, by optimally selecting the target agent $i \in \mathcal{T}$. Notably, the cost of implementing the intervention is fixed and does not depend on the specific agent targeted.  

This setup captures real-world scenarios where targeted interventions are used to mitigate systemic shocks. For example, during a pandemic, prioritizing vaccinations for healthcare workers can stabilize healthcare networks, reducing overall disruption. Similarly, in financial systems, stabilizing a key institution can prevent cascading defaults. This framework illustrates how targeting a single agent in a network can propagate benefits, minimizing aggregate losses and enhancing systemic resilience.  

We now analyze the government's targeting problem and derive the optimal intervention strategy. Fix a target $i \in \mathcal{T}$. Recall that $(I_n + \beta L)^{-2} \bm{1}_n = \bm{1}_n$, and for any eigenvector $\bm{q}_j$ of the Laplacian matrix $L$ with corresponding eigenvalue $\lambda_j$, we have $(I_n + \beta L)^{-2} \bm{q}_j = \omega_j \bm{q}_j$, where $\omega_j = \left( \frac{1}{1+\beta \lambda_j} \right)^2 \in (0,1]$.  

The aggregate social loss, $\bm{a}' \bm{a}$, can then be expressed as:  
\begin{equation}
\begin{aligned}\nonumber
\bm{a}' \bm{a} =&~ (\bm{1}_n - \bm{e}_i)' (I_n +\beta L)^{-2} (\bm{1}_n - \bm{e}_i)
\\
=&~ (n-1) -  \sum_{j=1}^{n-1} \left( 1- \omega_j \right) [q_j(i)]^2,
\end{aligned}
\end{equation}
where $q_j(i)$ is the $i$-th component of the Laplacian eigenvector $\bm{q}_j$. Here, $n-1$ represents the social loss when the government selects an agent, removes all links connected to that agent, and neutralizes the shock affecting that agent. The final term reflects the net gain from the network effect resulting from the targeting intervention.

From this, the government's targeting problem can be formulated as: 
\begin{equation}\label{eq:target_problem}
\max_{i \in \mathcal{T}}\, \varphi(i) \equiv \sum_{j=1}^{n-1} \left( 1- \omega_j \right) [q_j(i)]^2.
\end{equation}
This expression indicates that the optimal target is determined by the eigenstructure of the Laplacian and the weights $1 - \omega_j \in [0,1)$. Intuitively, targeting an agent with a high contribution to the weighted eigenvector centralities maximizes the reduction in social loss.

\begin{proposition}\label{prop_targeting}
The government selects the target for intervention based on the generalized Laplacian eigenvector centrality (gLEC) of each agent, weighted by $\phi_j = 1 - \left( \frac{1}{1+\beta \lambda_j} \right)^2$ for $j=1, \dots, n-1$.
\end{proposition}  

To provide intuition, we examine three contrasting scenarios:
\begin{description}
\item[Isolated agent (Degree $0$):]
Consider an agent $s \in \mathcal{T}$ who is completely isolated from the rest of the network. For such an agent, $q_j(s) > 0$ only if $\lambda_j = 0$, which implies $1 - \omega_j = 0$. Consequently, targeting an isolated agent results in a net gain of $\varphi(s) = 0$. Intuitively, the intervention fails to propagate benefits across the network, making the gain minimal and less impactful than in more connected structures.  
\item[Hub agent in a star network (Degree $n-1$):]
Consider a hub agent $h$ in a star network, where $h \in \mathcal{T}$ is fully connected to all other nodes. The eigenvector corresponding to the largest eigenvalue, $\lambda_1 = n$, satisfies $[q_1(h)]^2 = 1 - \frac{1}{n}$. This achieves the upper bound of the targeting problem in \eqref{eq:target_problem}, given by:
\begin{equation}\nonumber
	\varphi(h) = \left[ 1-\left(\frac{1}{1+\beta n} \right)^2 \right]\left(1-\frac{1}{n}\right).
\end{equation}
Targeting the hub maximizes the network-wide reduction in social loss due to the hub's dominant influence on equilibrium actions.
\item[Uniform intervention (Benchmark):] As a benchmark, consider a uniform intervention policy in which the government reduces the local state uniformly by $\frac{1}{n}$ across all agents, while ignoring implementation costs. The resulting equilibrium action vector is $\bm{a} = \left(1 - \frac{1}{n}\right) \bm{1}_n$, and the social loss is $n - 1 - \left(1 - \frac{1}{n}\right) \leq n-1-\varphi(h)$ with equality for $\beta = \infty$. This uniform intervention strategy achieves a greater reduction in social loss than targeting any single agent. Therefore, while the model highlights the benefits of targeted intervention, it also demonstrates its limitations in mitigating the impact of shocks in a networked society.
\end{description}
 
This analysis complements \cite{galeotti2020targeting}, who study optimal targeting interventions in networks using principal component decomposition of the adjacency matrix. While their model focuses on allocating resources across multiple agents to maximize utilitarian welfare under budget constraints (similar to the framework presented in \ref{sec:attenuating_shocks}), our approach focuses on single-agent targeting to mitigate the effects of a global shock. Both frameworks emphasize the role of network structure and eigenvalue decomposition, but the scope of interventions differs: 
\cite{galeotti2020targeting} consider interventions that continuously manipulate incentives across agents on a network, whereas our model focuses on a simple and intuitive selection problem, where gLEC serves as an index for optimal targeting. Extending this analysis to incorporate more detailed cost structures for targeting policies could yield additional insights into designing efficient interventions.

\section{Empirical application: Diffusion of microfinance}\label{sec:microfinance}

This section examines the role of centrality measures in the diffusion of microfinance information, using the dataset from \cite{banerjee2013diffusion}. The study investigates how social network positions influence the spread of information and participation in microfinance programs in 43 villages in Karnataka, India. Conducted in collaboration with Bharatha Swamukti Samsthe (BSS), a microfinance institution, the analysis highlights the effectiveness of leaders---such as teachers, shopkeepers, and social group leaders---as key information disseminators.

Prior to the introduction of the microfinance program, a comprehensive baseline survey was conducted in these villages to collect detailed household and social network data.%
\footnote{The survey aggregated 12 types of social interactions, including lending or borrowing money or goods and seeking advice, into a single undirected and unweighted network to capture the overall connectedness of households.} Using this network data, \cite{banerjee2013diffusion} introduced two novel measures, communication centrality and diffusion centrality, to quantify the effectiveness of leaders in spreading information. The microfinance program was implemented in these villages following the survey, allowing researchers to examine how leaders' centrality influenced both the dissemination of information and non-leaders' participation outcomes.%
\footnote{The baseline survey included 75 villages, but participation data is available only for the 43 BSS villages where the microfinance program was implemented.}

This paper extends the analysis by focusing on four centrality measures: degree, eigenvector, Katz-Bonacich (alpha), and proportional Laplacian Eigenvector Centrality (pLEC), with detailed descriptions provided in Appendix \ref{sec:a_variable}. In the main section, pLEC is computed by setting the LEC order proportional to the network size.%
\footnote{Specifically, the LEC order is set to 20\% of the network size. An alternative specification, using the cumulative sum of Laplacian eigenvalues (the 50\% rule), is examined in Appendix \ref{sec:a_plec_cumsum50pct}. Additional centrality measures, including diffusion centrality and communication centrality introduced by \cite{banerjee2013diffusion}, are also analyzed in Appendix \ref{sec:a_diffusion}.}
This section begins by examining the distributions of centrality scores across the sample villages, highlighting similarities and differences among these measures. It then presents a regression analysis to assess the impact of leaders' network positions on non-leaders' participation rates, as conducted in \cite{banerjee2013diffusion}.

\begin{figure}
    \centering
    \includegraphics[height = 0.45\linewidth]{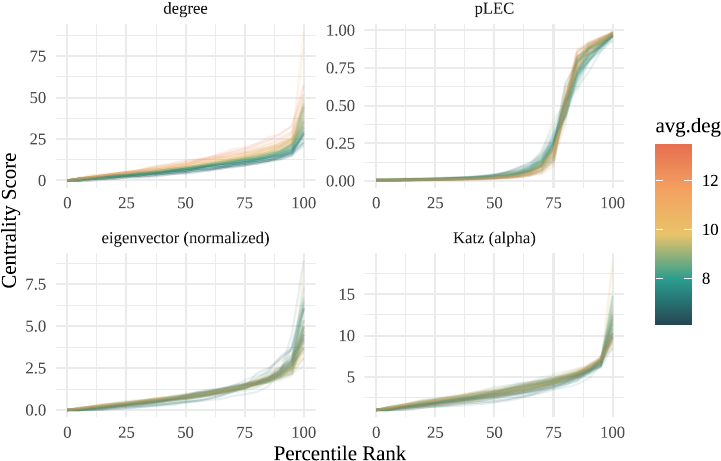}
\caption{Distributions of centrality scores (BSS).}\label{fig:dist_centr_BSS}
\end{figure}

Figure \ref{fig:dist_centr_BSS} presents the distributions of degree and other centrality measures for each village, plotted with percentile ranks on the x-axis and centrality values on the y-axis. For all centrality measures, the distributions of scores are remarkably consistent across the 43 villages, even when accounting for differences in average degree. This consistency highlights the structural regularity of the networks and indicates that these centrality measures provide comparable evaluations of node importance across networks of varying sizes and densities.

pLEC exhibits a distinct pattern compared to the other measures. It assigns very low scores to the bottom 60\% of nodes, followed by a sharp increase for the top 15--20\%, effectively capturing the network's core structure. In contrast, eigenvector centrality and Katz-Bonacich centrality display smoother distributions, with scores increasing almost linearly for the bottom 90\% of nodes and showing a sharp rise only at the top percentiles.

These differences in distribution suggest that pLEC is particularly effective at identifying a narrow subset of highly influential nodes in the network's core, while other centrality measures are better suited for capturing influence distributed more evenly across a broader range of nodes. These patterns illustrate how different centrality measures emphasize distinct structural features of the network, offering complementary perspectives on node importance.

\begin{table}
\centering
\caption{Eigenvector centrality with pLEC.}
\label{tab:regression_plec}
\begin{tabular}{lcccc}
\toprule
 & (1) & (2) & (3) & (4) \\
\midrule
log10(num\_hh) & -0.308*** & -0.322*** & -0.315*** & -0.347*** \\
               & (0.081)   & (0.082)   & (0.074)   & (0.080)   \\
len\_t         & 0.017*    & 0.018*    & 0.014*    &           \\
               & (0.006)   & (0.007)   & (0.007)   &           \\
eig\_normalized   &  0.084*   & 	       & 0.126***  & 0.157***  \\
               &  (0.034)  &           & (0.032)   & (0.031)   \\
plec\_ns20pct  &           & -0.138    & -0.353**  & -0.451*** \\
               &           & (0.130)   & (0.123)   & (0.125)   \\
Control        & yes       & yes       & yes       & yes       \\
\midrule
Num.Obs.       & 43        & 43        & 43        & 43        \\
R2             & 0.476     & 0.433     & 0.528     & 0.454     \\
R2 Adj.        & 0.371     & 0.320     & 0.416     & 0.345     \\
\bottomrule
\multicolumn{5}{l}{\footnotesize{+ p $<$ 0.1, * p $<$ 0.05, ** p $<$ 0.01, *** p $<$ 0.001}} \\
\end{tabular}
\end{table}

Following \cite{banerjee2013diffusion}, we analyze how leaders' network centrality influences non-leaders' microfinance participation rates using village-level regression analysis. Table \ref{tab:regression_plec} presents the results, where the dependent variable is the average participation rate among non-leader households. The independent variables include leaders' average eigenvector centrality (\textit{c\_eig\_normalized}) and pLEC scores (\textit{plec\_ns20pct}), as well as village characteristics such as network size (\textit{log10(num\_hh)}) and survey duration (\textit{len\_t}).%
\footnote{The other control variables include \textit{savings}, \textit{shgparticipate}, and \textit{fracGM\_survey}, as presented in Table \ref{tab:stats_BSS} in Appendix \ref{sec:a_summarystats}.}

Column (1) examines the effect of leaders' eigenvector centrality, controlling for village characteristics. Its positive and significant coefficient indicates that leaders with higher eigenvector centrality encourage greater participation among non-leaders.%
\footnote{In all columns, \textit{log10(num\_hh)} has a significant negative association with participation, while \textit{len\_t} correlates positively.}  
Column (2) replaces eigenvector centrality with pLEC (\textit{plec\_ns20pct}). The coefficient for pLEC is negative but not statistically significant, suggesting that, when considered independently, pLEC does not have a strong relationship with non-leader participation.

Columns (3) and (4) include both \textit{c\_eig\_normalized} and \textit{plec\_ns20pct}. In Column (3), the positive and significant effect of \textit{c\_eig\_normalized} becomes stronger compared to Column (1), while the negative effect of \textit{plec\_ns20pct} becomes both significant and larger in magnitude compared to Column (2). These results suggest that each centrality measure captures distinct aspects of leaders' network influence. Eigenvector centrality reflects leaders' ability to encourage participation through broader network connectivity, while pLEC appears to highlight structural constraints or bottlenecks that hinder non-leader engagement. Column (4), which excludes \textit{len\_t} as a control, further amplifies these patterns, with a larger positive coefficient for \textit{c\_eig\_normalized} and a more pronounced negative coefficient for \textit{plec\_ns20pct}.%
\footnote{The adjusted \( R^2 \) in Column (3) reaches 0.416, which is substantially higher than the adjusted \( R^2 \) from our replication of \cite{banerjee2013diffusion}. Details of this replication are provided in Appendix \ref{sec:a_centr_diffu}.}

Overall, the results highlight the contrasting roles of the two centrality measures. While eigenvector centrality captures leaders' capacity to spread influence, pLEC reflects the challenges associated with social ties, revealing a nontrivial relationship between network position and participation rates.%
\footnote{Robustness checks presented in Appendix \ref{sec:a_robustness} confirm the main findings, demonstrating consistency across alternative centrality measures, variations in pLEC specifications, and adjustments for network size and survey duration.}

\section{Conclusion}\label{sec:conclusion}

This paper introduces a new centrality measure, Laplacian Eigenvector Centrality (LEC), along with an analytical framework guided by the Laplacian spectrum. While LEC intuitively quantifies node positions and network structure with basic properties such as symmetry and periphery positions, it exhibits distinctive features compared to existing centrality measures and plays complementary roles in empirical research.

Using random graph models, we demonstrated the robustness and scalability of LEC, with practical guidance for selecting LEC orders proportional to network size. In particular, the analysis highlighted consistent scoring in scale-free and clustered networks.

The economic model and empirical application provided a distinctive perspective on social networks. While the literature extensively emphasizes the positive role of social connections in diffusion processes from an informational perspective, our findings highlight a contrasting aspect. Specifically, we show that social connections can generate coordination pressures, reducing responses to external shocks and adaptation to new technology. The regression analysis further suggested that such an incentive effect of social ties plays a crucial role in the dissemination of information and the design of targeting policies.

Future research can further explore LEC's applicability in specific social and economic networks, particularly those characterized by varying degree distributions and distinct incentive structures, thereby extending its relevance to broader settings in network science and policy design.

\bibliographystyle{ecca}

\clearpage
\appendix
\section*{Appendix}
\onehalfspacing
\renewcommand{\thefigure}{B\arabic{figure}}
\setcounter{figure}{0}
\renewcommand{\thetable}{B\arabic{table}}
\setcounter{table}{0}

\section{Omitted proofs}\label{sec:a_proof}
\begin{proof}[Proof of Property \ref{prop:symmetry}]
Let $G = (E,V)$ be an undirected graph of $n$ nodes and $L$ be the corresponding Laplacian matrix. Let $\mathcal{N}(i)$ denote the set of node $i$'s neighbors. That is, $\mathcal{N}(i) \equiv \{ j \in V: \{ij\} \in E \}$ or equivalently, $\mathcal{N}(i) \equiv \{ j \in V: A_{ij} = 1\}$.

Suppose without loss of generality that node 1 and 2 have the identical neighbors except for themselves. That is, $\mathcal{N}(1) \backslash \{2\} = \mathcal{N}(2) \backslash \{1\}$, or equivalently, $A_{1j} = A_{2j}$ for any $j \in \{3, 4, \dots, n\}$.

Let $\lambda$ and $\bm{q}$ be the Laplacian eigenvalue and its corresponding eigenvector. Then, by definition, we have $L \bm{q} = \lambda \bm{q}$. The first two components of this equation are expressed as
\begin{equation}\nonumber
\begin{aligned}
d_1 q(1) - A_{12} q(2) - \sum_{j = 3}^n A_{1j}q(j) = \lambda q(1)
\\
-A_{21}q(1) + d_2 q(2) - \sum_{j = 3}^n A_{2j}q(j) = \lambda q(2)
\end{aligned}
\end{equation}
where $d_1$ and $d_2$ denote the degree of node 1 and 2, $A_{12} = A_{21} \in \{0,1\}$ indicates the link between 1 and 2, and $q(j)$ denotes the $j$-th entry of vector $\bm{q}$. Note that $d_1 = d_2$ from the assumption of identical neighborhood.
Combining these two equations, we have 
\begin{equation}\label{eq:key_symmetry}
	(d_1+A_{21}-\lambda) (q(1) - q(2)) = 0.
\end{equation}

Suppose that for any order $k \in \{1,2,\dots, n-1\}$, $\lambda_k \neq d_1 + A_{21}$. Then, (\ref{eq:key_symmetry}) implies $q_k(1)=q_k(2)$ for all $k$. In this case, the LEC scores satisfy $c^{LEC(r)}_1 = c^{LEC(r)}_2$ for any LEC order $r$.

Suppose that for some $k$, we have $\lambda_k = d_1 + A_{21}$. Let $\underline{k}$ and $\bar{k}$ denote the smallest and largest index for such eigenvalues. We have three cases depending on the LEC order $r$.

\begin{description}
	\item Case 1: $r <\underline{k}$. \quad In this case, (\ref{eq:key_symmetry}) implies $q_{s}(1) = q_{s}(2)$ for $s \leq r$. Therefore, $c^{LEC(r)}_{1}  = \frac{1}{n}+\sum_{s = 1}^{r} q_{s}(1)^2 = \frac{1}{n}+\sum_{s = 1}^{r} q_{s}(2)^2 = c^{LEC(r)}_{2}$. 
	\item Case 2: $r>\bar{k}$. \quad In this case, (\ref{eq:key_symmetry}) implies $q_{s}(1) = q_{s}(2)$ for $s \geq r$. Therefore, $c^{LEC(r)}_{1}  = \frac{1}{n}+\sum_{s = 1}^{r} q_{s}(1)^2 = 1-\sum_{s = r+1}^{n-1} q_{s}(1)^2 = 1-\sum_{s = r+1}^{n-1} q_{s}(2)^2 = c^{LEC(r)}_{2}$. 
	\item Case 3: $r \in \left[\underline{k},\bar{k} \right]$.\quad In this case, we apply the modified LEC to take into account the eigenvalue multiplicity:
	\begin{equation}\nonumber
		\begin{aligned}
			c^{LEC(r)}_{1}  = \frac{1}{n}+\sum_{\tilde{r} = 1}^{\underline{k}-1} q_{s}(1)^2 
+ \frac{r-(\underline{k}-1)}{\bar{k}-(\underline{k}-1)} \sum_{s = \underline{k}}^{\bar{k}} q_{s}(1)^2 .
		\end{aligned}
	\end{equation}
	The proof completes as we check the following equation hold:
	\begin{equation}\nonumber
		\sum_{s = \underline{k}}^{\bar{k}} q_{s}(1)^2 = 1-\sum_{s \notin [ \underline{k},\bar{k}]} q_{s}(1)^2
		= 1-\sum_{s \notin [ \underline{k},\bar{k}]} q_{s}(2)^2
		=
		\sum_{s = \underline{k}}^{\bar{k}} q_{s}(2)^2 .
	\end{equation}
\end{description}
Hence, for any LEC order, the LEC scores for 1 and for 2 must be identical.
\end{proof}
\medskip

\begin{proof}[Proof of Property \ref{prop:terminal}]
Suppose that node 1 has degree 1 and is connected only to node 2.
Let $\lambda$ and $\bm{q}$ denote the Laplacian eigenvalue and its corresponding eigenvector.
The first component of the equation $L \bm {q} = \lambda \bm{q}$ is given by 
\begin{equation}\label{eq:terminal}
	q(1) - q(2) = \lambda q(1).
\end{equation}

\begin{description}
	\item For $0 \leq \lambda \leq 1$, (\ref{eq:terminal}) implies that $(1-\lambda) q(1) = q(2)$. Then, either $0 \leq q(2) \leq q(1)$ or $q(1) \leq q(2) \leq 0$ holds.
	\item For $1< \lambda \leq 2$, (\ref{eq:terminal}) implies that $(\lambda-1) q(1) = -q(2)$. Then, either $0 \leq -q(2) \leq q(1)$ or $q(1) \leq -q(2) \leq 0$ holds.
	\item For $2< \lambda$, (\ref{eq:terminal}) implies that $(\lambda-1) q(1) = -q(2)$. Then, either $0 \leq q(1) \leq -q(2)$ or $-q(2) \leq q(1) \leq 0$ holds.
\end{description}
Hence as long as $\lambda > 2$, $q(2)^2 \geq q(1)^2$. 

Now fix a LEC order $r$. If $\lambda_r >2$, then 
\begin{equation}\nonumber
	c^{LEC(r)}_2 = \frac{1}{n} + \sum_{k=1}^r q_k(2)^2 \geq \frac{1}{n} + \sum_{k=1}^r q_k(1)^2 =c^{LEC(r)}_1.
\end{equation}
If $\lambda_r \leq 2$, then 
\begin{equation}\nonumber
	c^{LEC(r)}_2 = 1-\sum_{k=r+1}^{n-1} q_k(2)^2 \geq  1-\sum_{k=r+1}^{n-1} q_k(1)^2  =c^{LEC(r)}_1.
\end{equation}
Thus, for any LEC order $r$, the LEC score for node 2 is greater than or equal to that for node 1.
\end{proof}
\medskip

\begin{proof}[Proof of Property \ref{prop:isolated}]
Suppose that $d_j=0$ for $j \in \{1,2,\dots, k\}$. Let $\bm{e}_i \in \mathbb{R}^n$ denote the vector such that only $i$-th entry is 1 and 0 otherwise. Then, for any $j \in \{1,2, \dots, k\}$, we have $L \bm{e}_i = 0\bm{e}_i$. Since $\{\bm{e}_1, \bm{e}_2, \dots, \bm{e}_k, \bm{1}_n \}$ are linearly independent, the Laplacian eigenvalues satisfy $\lambda_{n-j}=0$ for $j \in \{1,2,\dots, k\}$. Let $S \subset \mathbb{R}^n$ be the subspace spanned by $\{\bm{e}_1, \bm{e}_2, \dots, \bm{e}_k, \bm{1}_n \}$. Then, for any $\bm{q} \in \mathbb{R}^n \setminus M$, we have $q(j) = 0$ for any $j \in \{1,2, \dots, k\}$. This implies that for any LEC order $r \leq n-k-1$, $c^{LEC(r)}_j = q_0(j)^2 = \frac{1}{n}$.
\end{proof}
\medskip

\begin{proof}[Proof of Property \ref{prop:core}]
Suppose that for $i = 1,2, \dots, k$, the degree is $d_i=n-1$ and for $i=k+1, k+2, \dots, n$, the degree is $d_i = k$.

Let $\bm{v}_i \equiv n \bm{e}_i -\bm{1}_n$ where $\bm{e}_i$ is the vector such that only $i$-th entry is 1 and 0 otherwise. Then, we have $L\bm{v}_i = n \bm{v}_i$.
Also, for any $\bm{y}_{n-k} \in \mathbb{R}^{n-k}$ orthogonal to $\bm{1}_{n-k}$, we have 
\begin{equation}\nonumber
	L \begin{bmatrix}
\bm{0}_k
\\
\bm{y}_{n-1}
\end{bmatrix}
= k
\begin{bmatrix}
\bm{0}_k
\\
\bm{y}_{n-k}
\end{bmatrix}.
\end{equation}
Thus, there are three types of eigenvectors $\bm{1}_n$, $\bm{v}_i$ and $(\bm{0}_k',\bm{y}_{n-1}')'$ with corresponding eigenvalues $0$, $n$ and $k$, respectively.

Thus, the LEC score for the hubs is given by 
\begin{equation}\nonumber
	c^{LEC(r)}_{hub} = \begin{cases}
\frac{1}{n} + \frac{(n-1)r}{nk}&\text{for }r \in\{ 0, 1, \dots, k\}
\\
1 &\text{for }r \in\{ k+1, \dots, n-1\}
\end{cases}
\end{equation}

The LEC score for the peripheries is 
\begin{equation}\nonumber
	c^{LEC(r)}_{per} = \begin{cases}
\frac{1}{n} + \frac{r}{n(n-k)}&\text{for }r \in\{ 0, 1, \dots, k\}
\\
\frac{1}{n}+\frac{k}{n(n-k)}+\frac{r-k}{n-k} &\text{for }r \in \{ k+1, \dots, n-1\}
\end{cases}
\end{equation}
This completes the proof.
\end{proof}
\medskip

\begin{proof}[Proof of Proposition \ref{prop_targeting}]
Let $Q$ be the matrix of eigenvectors of $L$, and let $\Omega$ be a diagonal matrix with entries $\omega_j =  \left( \frac{1}{1+\beta \lambda_j} \right)^2$, the eigenvalues of $(I_n +\beta L)^{-2}$. Then, the social loss can be expressed as:
\begin{equation}
\begin{aligned}\nonumber
\bm{a}' \bm{a} =&~ (\bm{1}_n - \bm{e}_i)' (I_n +\beta L)^{-2} (\bm{1}_n - \bm{e}_i)
\\
=&~ n-2+ \bm{e}_i' Q \Omega Q' \bm{e}_i
\\
=&~ n-2+ \mathrm{tr}(\Omega Q' \bm{e}_i\bm{e}_i' Q)
\end{aligned}
\end{equation}

Next, compute $Q'\bm{e}_i$:
\begin{equation}
\begin{aligned}\nonumber
Q' \bm{e}_i =&~ 
\begin{bmatrix}
	\bm{q}_0' \\ 
	\bm{q}_1' \\ 
	\vdots \\
	\bm{q}_{n-1}' \\ 
\end{bmatrix} \bm{e}_i
= 
\begin{bmatrix}
q_0(i) \\ q_1(i) \\ \vdots \\ q_{n-1}(i)
\end{bmatrix}.
\end{aligned}
\end{equation}
Hence, 
\begin{equation}\nonumber
\mathrm{tr}(\Omega Q' \bm{e}_i\bm{e}_i' Q) = \sum_{j=0}^{n-1} \omega_j [q_j(i)]^2.
\end{equation}
Finally, using the orthonormality property $\sum_{j=0}^{n-1} [q_j(i)]^2=1$, we obtain the formulation presented in \eqref{eq:target_problem}.
\end{proof}

\renewcommand{\thefigure}{B\arabic{figure}}
\setcounter{figure}{0}
\renewcommand{\thetable}{B\arabic{table}}
\setcounter{table}{0}
\setcounter{section}{1}

\section{Supplementary materials for microfinance application}

\subsection{Descriptive and comparative analyses}
\subsubsection{Centrality measures and normalization}\label{sec:a_variable}

This subsection describes the centrality measures analyzed in the empirical application and outlines the normalization procedures applied to ensure comparability across villages.

In the main section, \textit{pLEC} is computed using the LEC when the LEC order is set to 20\% of the network size, denoted as \textit{plec\_ns20pct} in the Appendix. This choice is justified by the discussion in Section \ref{sec:statistical} and the distribution of Laplacian eigenvalues shown in Figure \ref{fig:lambda_all_BSS}. For robustness, an alternative measure, \textit{plec\_cum50pct}, which sets the LEC order based on the cumulative sum of Laplacian eigenvalues reaching 50\%, is also computed and analyzed in Appendix \ref{sec:a_robustness}.

The eigenvector centrality, \textit{c\_eig}, is defined as the eigenvector of the adjacency matrix corresponding to the largest eigenvalue, including households with zero degrees.%
\footnote{In \cite{banerjee2013diffusion}, the leader averages of centrality scores are computed after dropping households with zero degrees from the dataset. This approach can distort the leader average, as a leader with one neighbor who becomes zero degree would increase the average centrality score. To avoid such distortions, zero-degree leaders are included when computing leader averages.}  
However, the leader averages of \textit{c\_eig} exhibit a strong negative correlation with network size (correlation coefficient of -0.724). To address this issue, a normalization procedure is applied to derive \textit{c\_eig\_normalized}, ensuring that the mean eigenvector centrality within each village remains constant and independent of network size (see Appendix \ref{sec:a_eig_normalization} for further details).

\textit{c\_bona\_alpha} (or \textit{Katz} in the main section) corresponds to Katz-Bonacich (alpha) centrality, while \textit{c\_bona\_power} represents Bonacich power centrality (also referred to as beta centrality). For both measures, the decay factors are set to $0.8 / \mu$, where $\mu$ is the largest eigenvalue of the adjacency matrix. Following the practice in \cite{banerjee2013diffusion}, this adjustment ensures that the distribution of centrality scores is independent of network size and density across villages.

The diffusion centralities, denoted as \textit{c\_diffu\_}, are calculated for varying numbers of iterations. \textit{c\_diffu\_t} uses the duration of the survey (\textit{len\_t}) as the number of iterations, as in \cite{banerjee2013diffusion}. Additionally, \textit{c\_diffu\_03}, \textit{c\_diffu\_07}, and \textit{c\_diffu\_10} are computed with fixed iterations of $T = 3$, $T = 7$, and $T = 10$, respectively.

\subsubsection{Detailed Summary Statistics of Village Networks}\label{sec:a_summarystats}

\begin{table}[h]

\caption{\label{tab:stats_BSS}Summary statistics: Village level variables (BSS)}
\centering
\begin{tabular}[t]{lrrrrr}
\toprule
variable & n & mean & sd & min & max\\
\midrule
num\_hh & 43 & 223.209 & 56.170 & 114.000 & 356.000\\
avg\_deg & 43 & 9.196 & 1.653 & 6.128 & 12.781\\
fractionLeaders & 43 & 0.121 & 0.031 & 0.064 & 0.193\\
len\_t & 43 & 6.558 & 1.830 & 2.000 & 10.000\\
\addlinespace
mf\_leader & 43 & 0.248 & 0.125 & 0.036 & 0.556\\
mf\_nonleader & 43 & 0.185 & 0.084 & 0.068 & 0.438\\
degree\_leader & 43 & 12.729 & 2.559 & 8.538 & 18.818\\
degree\_nonleader & 43 & 8.715 & 1.603 & 5.576 & 12.204\\
\addlinespace
plec\_ns20pct & 43 & 0.342 & 0.064 & 0.215 & 0.505\\
plec\_cum50pct & 43 & 0.359 & 0.064 & 0.245 & 0.535\\
c\_eig & 43 & 0.073 & 0.017 & 0.043 & 0.123\\
c\_eig\_normalized & 43 & 1.441 & 0.244 & 0.994 & 2.216\\
c\_bona\_alpha & 43 & 4.292 & 0.448 & 3.311 & 5.191\\
c\_bona\_power & 43 & 1.099 & 0.141 & 0.804 & 1.501\\
\addlinespace
c\_diffu\_t & 43 & 5.410 & 1.777 & 2.064 & 10.503\\
c\_diffu\_03 & 43 & 2.512 & 0.294 & 1.931 & 3.111\\
c\_diffu\_07 & 43 & 5.759 & 0.791 & 3.996 & 7.344\\
c\_diffu\_10 & 43 & 8.166 & 1.192 & 5.458 & 10.534\\
\addlinespace
savings & 43 & 1.613 & 0.098 & 1.354 & 1.837\\
shgparticipate & 43 & 0.207 & 0.084 & 0.014 & 0.354\\
fracGM\_survey & 43 & 2.509 & 0.373 & 1.068 & 3.030\\
\bottomrule
\multicolumn{6}{l}{\textsuperscript{} \footnotesize{\textit{Notes}: The statistics displayed do not incorporate the bias correction}}\\
\multicolumn{6}{l}{\footnotesize{proposed by \cite{chandrasekhar2011econometrics}.}}\\
\end{tabular}
\end{table}

Table \ref{tab:stats_BSS} presents the descriptive statistics of the village-level variables from the 43 BSS sample villages used in the regression analysis. The number of households ranges from 114 to 356, with an average of 223.2. The average degree across villages varies between 6.1 and 12.8. On average, 12.1\% of households have members designated as ``leaders.''%
\footnote{The fraction of leaders shows a weak negative correlation with the number of households, with a correlation coefficient of -0.25.}  
The variable \textit{len\_t}, representing the survey duration for each village, ranges from 2 to 10 periods, with a mean of 6.56.%
\footnote{In Appendix \ref{sec:a_centr_diffu}, we explore the effects of survey duration on participation rates and assess the validity of the diffusion centrality measure used in \cite{banerjee2013diffusion}.}  
The average participation rates in microfinance (\textit{mf}) are 24.8\% for leaders and 18.5\% for non-leaders. Leaders also have a higher average degree of 12.7, compared to 8.7 for non-leaders.  

In the regression analysis, we investigate how leaders' centrality affects non-leaders' participation in microfinance. The centrality measures listed in Table \ref{tab:stats_BSS} include pLEC, eigenvector centrality, Bonacich centrality, and diffusion centrality.

\subsubsection{Distributions of Laplacian Eigenvalues}

\begin{figure}[h]\centering
\includegraphics[width=0.65\textwidth]{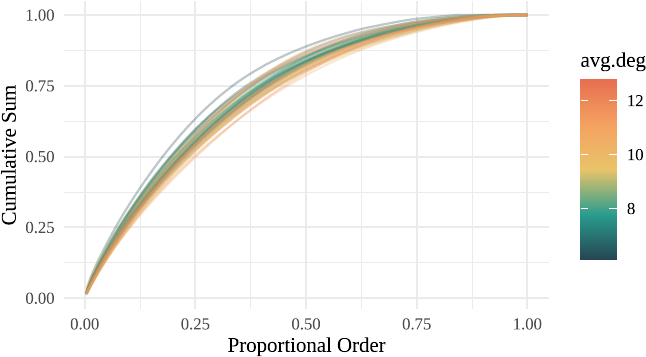}
\caption{Cumulative sum of Laplacian eigenvalues (BSS).}\label{fig:lambda_all_BSS}
\end{figure}

Figure \ref{fig:lambda_all_BSS} illustrates the cumulative sum of Laplacian eigenvalues normalized by their total sum for each village, with color gradients indicating the average degree of each village network, as introduced in Section \ref{sec:statistical}.%
\footnote{The x-axis represents the proportion of the network size used to determine the LEC order, while the y-axis reflects the cumulative proportion of eigenvalues captured.}  
The variation across villages, as represented by the colored curves, suggests that the distribution of Laplacian eigenvalues is broadly consistent, regardless of differences in network size and network density (proxied by average degree).  

Notably, selecting an LEC order equivalent to 20\% of the network size captures approximately 50\% of the total sum of Laplacian eigenvalues across the sample. This result supports the use of proportional LEC (pLEC) orders, as it strikes a balance between generating sufficient variation in LEC scores among households within each village and reflecting the core structural features of the networks.

\subsubsection{Correlation between centrality measures}

\begin{figure}[h]
    \centering
    \includegraphics[width=0.9\linewidth]{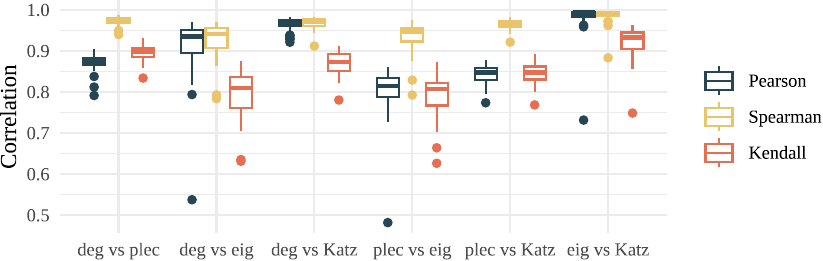} 
    \caption{Box plot of pairwise correlation between centrality measures (BSS).}
    \label{fig:correlation_boxplot_BSS}
\end{figure}

Figure~\ref{fig:correlation_boxplot_BSS} presents the box plot of pairwise correlation coefficients (Pearson, Spearman, and Kendall) between centrality measures across the BSS sample villages. The figure illustrates six pairwise comparisons: three involve degree centrality with other measures (eigenvector, Katz-Bonacich, and pLEC), two involve pLEC with other measures (eigenvector and Katz-Bonacich), and the rightmost boxplots represent the comparison between eigenvector and Katz-Bonacich centrality. Each boxplot reflects the range and variability of correlation coefficients across the BSS villages for the respective pair of centrality measures.

Focusing on the first three sets of boxplots, we observe distinct patterns in the correlations involving degree centrality. The correlations between degree and Katz-Bonacich centrality are consistently high across all villages, with Pearson coefficients around 0.95 and limited variability. However, Kendall's rank correlations are slightly lower and exhibit greater variability. In contrast, the correlations between degree and eigenvector centrality show greater variability, with outliers and a noticeable gap between Pearson and Kendall coefficients, highlighting the localization effect of eigenvector centrality. For degree and pLEC centralities, Kendall's rank correlations stand out, averaging around 0.90 with a narrow range, reflecting pLEC's tendency to align with degree rankings while capturing additional structural features.

These observations are broadly consistent with the patterns observed in the boxplot generated from the Erd\H{o}s--R\'enyi (ER) model (Figure~\ref{fig:boxplot_correlation}). However, notable differences emerge. In the BSS sample, the correlations between degree and pLEC are systematically lower compared to those observed in the ER model, suggesting that the structural heterogeneity and clustering in real-world networks reduce the alignment between pLEC and degree centrality.

Turning to the correlations between pLEC and eigenvector or Katz-Bonacich centrality, distinct patterns emerge. A shared feature of these comparisons is the small gap between Pearson and Kendall coefficients, indicating that pLEC rankings are reasonably well-aligned with those of eigenvector and Katz-Bonacich centrality. However, the correlations between pLEC and eigenvector centrality exhibit a lower average and wider variability, with coefficients typically ranging from 0.70 to 0.85 and occasional downward outliers, reflecting inconsistent alignment across villages. By contrast, the correlations between pLEC and Katz-Bonacich centrality are higher and more stable, concentrating around 0.80 to 0.90, which suggests a stronger and more consistent relationship.

The rightmost boxplots display the correlations between eigenvector and Katz-Bonacich centrality. Pearson and Spearman coefficients are consistently high, generally above 0.95, with minimal variability across villages. Kendall's rank correlations, while slightly lower, remain above 0.90, reflecting strong alignment in node rankings between the two measures. These results suggest that eigenvector and Katz-Bonacich centrality capture highly similar structural information.

\subsection{Robustness and technical discussions}
\label{sec:a_robustness}

The robustness checks in the appendix validate and extend the findings from the main analysis. Alternative specifications of proportional Laplacian Eigenvector Centrality (pLEC), including variations based on cumulative eigenvalue proportions, yield consistent results, reaffirming pLEC's unique role in capturing network bottlenecks. Similarly, incorporating classical centrality measures, such as Bonacich power and Katz-Bonacich alpha centralities, alongside pLEC, underscores their complementary explanatory power in predicting non-leader participation.

\subsubsection{Alternative pLEC measure}\label{sec:a_plec_cumsum50pct}

\begin{table}[h]
\centering
\caption{Eigenvector centrality with pLEC based on the cumulative sum of eigenvalues.}\label{tab:reg_plec_cum50pct}
\centering
\begin{tabular}[t]{lcccc}
\toprule
  & (1) & (2) & (3) & (4)\\
\midrule
log10(num\_hh) & \num{-0.308}*** & \num{-0.326}*** & \num{-0.320}*** & \num{-0.352}***\\
 & (\num{0.081}) & (\num{0.081}) & (\num{0.074}) & (\num{0.080})\\
len\_t & \num{0.017}* & \num{0.018}* & \num{0.014}+ & \\
 & (\num{0.006}) & (\num{0.007}) & (\num{0.007}) & \\
c\_eig\_normalized & \num{0.084}* &  & \num{0.116}*** & \num{0.144}***\\
 & (\num{0.034}) &  & (\num{0.031}) & (\num{0.029})\\
plec\_cum50pct &  & \num{-0.171} & \num{-0.328}* & \num{-0.429}**\\
 &  & (\num{0.127}) & (\num{0.121}) & (\num{0.122})\\
Control & yes & yes & yes & yes\\
\midrule
Num.Obs. & \num{43} & \num{43} & \num{43} & \num{43}\\
R2 & \num{0.476} & \num{0.439} & \num{0.526} & \num{0.458}\\
R2 Adj. & \num{0.371} & \num{0.327} & \num{0.415} & \num{0.350}\\
\bottomrule
\multicolumn{5}{l}{\rule{0pt}{1em}+ p $<$ 0.1, * p $<$ 0.05, ** p $<$ 0.01, *** p $<$ 0.001}\\
\end{tabular}
\end{table}

Table \ref{tab:reg_plec_cum50pct} examines the robustness of the results to an alternative specification of the proportional Laplacian Eigenvector Centrality (pLEC). Instead of setting the pLEC order as 20\% of the network size, this alternative measure determines the LEC order based on the cumulative proportion of Laplacian eigenvalues, selecting the smallest order such that the cumulative sum of eigenvalues reaches 50\% of the total sum. This approach adjusts the pLEC calculation to account for the spectral properties of individual village networks.

The regression results in Table \ref{tab:reg_plec_cum50pct} maintain the same specification as the main regressions in the text. The coefficients for eigenvector centrality (\textit{c\_eig\_normalized}) remain positive and statistically significant, reinforcing the finding that leaders' influence, as captured by eigenvector centrality, promotes non-leader participation. Conversely, the alternative pLEC measure shows a negative and significant relationship with non-leader participation, similar to the primary pLEC specification. This consistency across pLEC measures highlights the robustness of the observed contrasting roles of eigenvector centrality and pLEC.

\subsubsection{Robustness across alternative measures}\label{sec:a_diffusion}

\begin{table}[h]
\centering
\caption{Alternative centrality measures with pLEC.}\label{tab:alt_centrality}
\centering
\begin{tabular}[t]{lcccccc}
\toprule
  & (1) & (2) & (3) & (4) & (5) & (6)\\
\midrule
log10(num\_hh) & \num{-0.294}*** & \num{-0.297}*** & \num{-0.294}*** & \num{-0.269}** & \num{-0.292}*** & \num{-0.270}**\\
 & (\num{0.080}) & (\num{0.081}) & (\num{0.080}) & (\num{0.078}) & (\num{0.072}) & (\num{0.078})\\
len\_t & \num{0.019}** & \num{0.018}** & \num{0.019}** & \num{0.018}* & \num{0.015}* & \num{0.018}*\\
 & (\num{0.007}) & (\num{0.006}) & (\num{0.007}) & (\num{0.007}) & (\num{0.007}) & (\num{0.007})\\
c\_bona\_alpha & \num{0.024} &  &  & \num{0.067}* &  & \\
 & (\num{0.023}) &  &  & (\num{0.029}) &  & \\
c\_bona\_power &  & \num{0.123}+ &  &  & \num{0.244}*** & \\
 &  & (\num{0.066}) &  &  & (\num{0.064}) & \\
c\_diffu\_07 &  &  & \num{0.014} &  &  & \num{0.038}*\\
 &  &  & (\num{0.013}) &  &  & (\num{0.016})\\
plec\_ns20pct &  &  &  & \num{-0.431}* & \num{-0.450}** & \num{-0.426}*\\
 &  &  &  & (\num{0.181}) & (\num{0.146}) & (\num{0.179})\\
Control & yes & yes & yes & yes & yes & yes\\
\midrule
Num.Obs. & \num{43} & \num{43} & \num{43} & \num{43} & \num{43} & \num{43}\\
R2 & \num{0.438} & \num{0.461} & \num{0.438} & \num{0.493} & \num{0.531} & \num{0.493}\\
R2 Adj. & \num{0.325} & \num{0.354} & \num{0.326} & \num{0.373} & \num{0.421} & \num{0.373}\\
\bottomrule
\multicolumn{7}{l}{\rule{0pt}{1em}+ p $<$ 0.1, * p $<$ 0.05, ** p $<$ 0.01, *** p $<$ 0.001}\\
\end{tabular}
\end{table}

Table \ref{tab:alt_centrality} examines the robustness of the relationship between leaders' network centrality and non-leaders' participation rates, incorporating Katz-Bonacich alpha centrality, Bonacich power centrality, and diffusion centrality with fixed iterations (\textit{c\_diffu\_07}) alongside proportional LEC (\textit{plec\_ns20pct}).

The alternative centrality measures, including Katz-Bonacich alpha, Bonacich power, and diffusion centrality, display positive and significant effects on participation, consistent with the earlier results using eigenvector centrality. Across all specifications, \textit{plec\_ns20pct} consistently shows a negative and significant association with non-leaders' participation rates (Columns 4 and 6), highlighting its role in capturing structural constraints or bottlenecks.

The adjusted \( R^2 \) values are similar across specifications, with Bonacich power centrality achieving the highest adjusted \( R^2 \) of 0.421 in Column (5). These results confirm that pLEC and alternative centrality measures capture distinct yet complementary aspects of leaders' influence on participation, reinforcing the robustness of the main findings.


\subsubsection{Assessing the validity of diffusion centrality}\label{sec:a_centr_diffu}

This section evaluates the validity of diffusion and communication centralities introduced by \cite{banerjee2013diffusion} by examining their relationship with survey duration (\textit{len\_t}) and their robustness under alternative specifications. The findings highlight potential limitations in relying on centrality measures that are strongly correlated with non-structural variables like survey duration. 

\paragraph{Correlation with survey duration}~

\begin{figure}[h]
\centering
\subfloat[len\_t vs.~diffusion centrality.]{\includegraphics[width=0.475\linewidth]{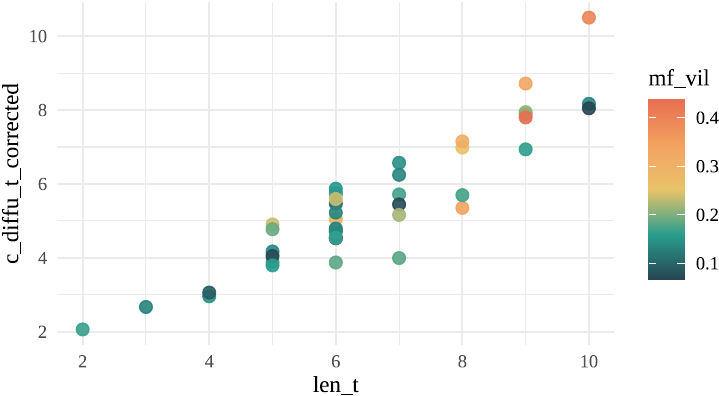}\label{fig:plot_diffu}}\hspace{7mm}
\subfloat[len\_t vs.~communication centrality.]{\includegraphics[width=0.475\linewidth]{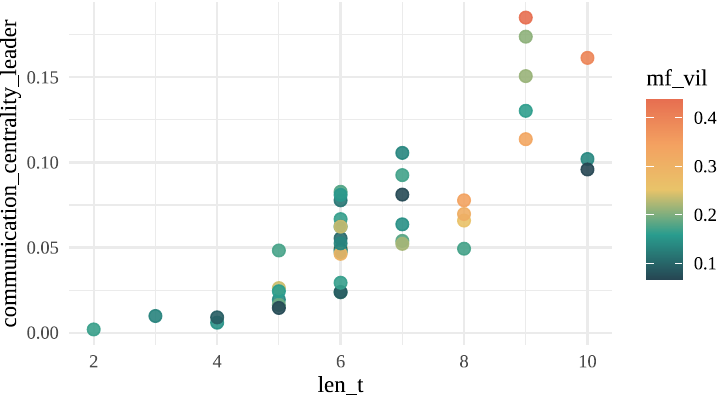}\label{fig:plot_commu}}
\caption{Survey duration and diffusion/communication centralities.}\label{fig:plot_lent}
\end{figure}%

Figures \ref{fig:plot_diffu} and \ref{fig:plot_commu} depict the relationship between survey duration (\textit{len\_t}) and the centrality measures proposed in \cite{banerjee2013diffusion}: diffusion centrality (\textit{c\_diffu\_t\_corrected}) and communication centrality (\textit{communication\_centrality\_leader}). The figures reveal a strong positive correlation between survey duration and these centralities, consistent with the iterative construction of \textit{c\_diffu\_t\_corrected}, which depends directly on the number of diffusion iterations set equal to \textit{len\_t}. This suggests that communication centrality may also indirectly reflect survey duration rather than purely structural aspects of the network. Furthermore, the color gradient, representing non-leaders' microfinance participation rates (\textit{mf\_vil}), indicates a positive association between participation and survey duration, complicating the interpretation of these centralities.

\paragraph{Replication of core findings}~

\begin{table}[h]
\centering
\caption{Replication of Fig.2 in \cite{banerjee2013diffusion}.}\label{tab:replicate}
\centering
\begin{tabular}[t]{lccccc}
\toprule
  & (1) & (2) & (3) & (4) & (5)\\
\midrule
num\_hh & \num{-0.001}** & \num{-0.001}** & \num{0.000} & \num{0.000} & \num{0.000}+\\
 & (\num{0.000}) & (\num{0.000}) & (\num{0.001}) & (\num{0.001}) & (\num{0.000})\\
communication\_centrality\_leader & \num{0.766}* &  & \num{0.713} &  & \\
 & (\num{0.335}) &  & (\num{0.428}) &  & \\
c\_diffu\_t\_corrected &  & \num{0.022}** &  & \num{0.018}+ & \num{0.020}**\\
 &  & (\num{0.007}) &  & (\num{0.009}) & (\num{0.007})\\
c\_eig\_corrected &  &  & \num{3.572} & \num{3.691} & \num{0.612}\\
 &  &  & (\num{2.330}) & (\num{2.265}) & (\num{0.854})\\
Other Centralities & no & no & yes & yes & no\\
Control & yes & yes & yes & yes & yes\\
\midrule\\
Num.Obs. & \num{43} & \num{43} & \num{43} & \num{43} & \num{43}\\
R2 & \num{0.406} & \num{0.442} & \num{0.530} & \num{0.515} & \num{0.448}\\
R2 Adj. & \num{0.307} & \num{0.349} & \num{0.342} & \num{0.321} & \num{0.338}\\
\bottomrule
\multicolumn{6}{l}{\rule{0pt}{1em}+ p $<$ 0.1, * p $<$ 0.05, ** p $<$ 0.01, *** p $<$ 0.001}\\
\end{tabular}
\end{table}

Table \ref{tab:replicate} replicates the main results from \cite{banerjee2013diffusion}, using diffusion and communication centralities as predictors. The reported adjusted $R^2$ values range from 0.307 to 0.349, which are substantially lower than those obtained in our main analysis (e.g., Table \ref{tab:regression_plec} and Table \ref{tab:alt_centrality}). This discrepancy highlights the limited explanatory power of these centrality measures relative to alternatives such as eigenvector and proportional Laplacian Eigenvector Centralities (pLEC).

\paragraph{Incorporating survey duration}~

\begin{table}[h]
\centering
\caption{Fig.2 in \cite{banerjee2013diffusion} with survey duration.}\label{tab:replicate_with_len_t}
\centering
\begin{tabular}[t]{lccccc}
\toprule
  & (1) & (2) & (3) & (4) & (5)\\
\midrule
num\_hh & \num{-0.001}** & \num{-0.001}** & \num{0.000} & \num{-0.001} & \num{0.000}\\
 & (\num{0.000}) & (\num{0.000}) & (\num{0.001}) & (\num{0.001}) & (\num{0.000})\\
len\_t & \num{0.011} & \num{-0.001} & \num{-0.003} & \num{-0.030} & \num{0.024}\\
 & (\num{0.013}) & (\num{0.015}) & (\num{0.022}) & (\num{0.040}) & (\num{0.025})\\
communication\_centrality\_leader & \num{0.393} &  & \num{0.827} &  & \\
 & (\num{0.674}) &  & (\num{1.159}) &  & \\
c\_diffu\_t\_corrected &  & \num{0.023} &  & \num{0.051} & \num{-0.007}\\
 &  & (\num{0.016}) &  & (\num{0.046}) & (\num{0.029})\\
c\_eig\_corrected &  &  & \num{3.599} & \num{3.889} & \num{1.906}\\
 &  &  & (\num{2.359}) & (\num{2.364}) & (\num{1.346})\\
Other Centralities & no & no & yes & yes & no\\
Control & yes & yes & yes & yes & yes\\
\midrule\\
Num.Obs. & \num{43} & \num{43} & \num{43} & \num{43} & \num{43}\\
R2 & \num{0.421} & \num{0.442} & \num{0.531} & \num{0.522} & \num{0.459}\\
R2 Adj. & \num{0.306} & \num{0.330} & \num{0.320} & \num{0.307} & \num{0.332}\\
\bottomrule
\multicolumn{6}{l}{\rule{0pt}{1em}+ p $<$ 0.1, * p $<$ 0.05, ** p $<$ 0.01, *** p $<$ 0.001}\\
\end{tabular}
\end{table}

Table \ref{tab:replicate_with_len_t} introduces survey duration as a control variable to further evaluate the validity of diffusion and communication centralities. When \textit{len\_t} is included in the regression model, the coefficients for these centrality measures become highly sensitive to additional controls and lose statistical significance. This sensitivity underscores the potential confounding role of survey duration, as centralities derived from iterative diffusion processes may capture time-dependent effects rather than intrinsic structural features of the network.%
\footnote{Log-transforming the household count variable improves the model fit, as evidenced by higher adjusted $R^2$ values, without altering the qualitative findings. Detailed results are available upon request.}


\subsubsection{Eigenvector centrality with and without normalization}\label{sec:a_eig_normalization}

This subsection examines the implications of normalization and correction applied to eigenvector centrality measures, focusing on their impact on cross-network comparability and robustness in regression models.

\paragraph{Centrality scores and size effects}~

\begin{figure}[h]
\centering
\includegraphics[height = 0.45\linewidth]{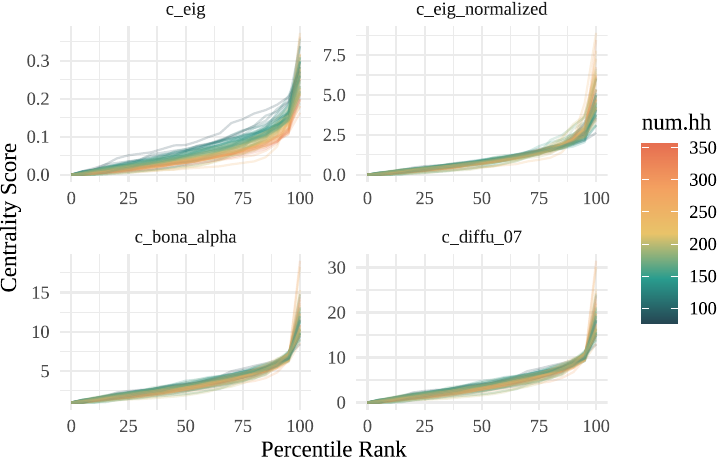}
\caption{Distribution of centrality scores: Eigenvector centrality}\label{fig:dist_centr_eig}
\end{figure}

Figure \ref{fig:dist_centr_eig} displays the distributions of four eigenvector-based measures: unadjusted eigenvector centrality (\textit{c\_eig}), normalized eigenvector centrality (\textit{c\_eig\_normalized}), Katz-Bonacich alpha centrality (\textit{c\_bona\_alpha}), and diffusion centrality (\textit{c\_diffu\_07}). While \textit{c\_eig} exhibits substantial variation across networks, particularly for nodes with low scores, the normalized version (\textit{c\_eig\_normalized}) shows a smoother distribution across percentiles, with reduced dependence on network size. 

\paragraph{Regression with and without normalization}~

\begin{table}[h]
\centering
\caption{Eigenvector centrality (with and without normalization) with pLEC}\label{tab:reg_eig}
\centering
\begin{tabular}[t]{lcccccc}
\toprule
  & (1) & (2) & (3) & (4) & (5) & (6)\\
\midrule
log10(num\_hh) & \num{-0.183}+ & \num{-0.159} & \num{-0.308}*** & \num{-0.086} & \num{-0.071} & \num{-0.315}***\\
 & (\num{0.103}) & (\num{0.110}) & (\num{0.081}) & (\num{0.102}) & (\num{0.117}) & (\num{0.074})\\
len\_t & \num{0.018}** & \num{0.018}** & \num{0.017}* & \num{0.016}* & \num{0.016}* & \num{0.014}*\\
 & (\num{0.006}) & (\num{0.006}) & (\num{0.006}) & (\num{0.007}) & (\num{0.007}) & (\num{0.007})\\
c\_eig & \num{1.311}+ &  &  & \num{2.379}** &  & \\
 & (\num{0.740}) &  &  & (\num{0.754}) &  & \\
c\_eig\_corrected &  & \num{1.461}+ &  &  & \num{2.389}** & \\
 &  & (\num{0.785}) &  &  & (\num{0.819}) & \\
c\_eig\_normalized &  &  & \num{0.084}* &  &  & \num{0.126}***\\
 &  &  & (\num{0.034}) &  &  & (\num{0.032})\\
plec\_ns20pct &  &  &  & \num{-0.389}** & \num{-0.351}* & \num{-0.353}**\\
 &  &  &  & (\num{0.140}) & (\num{0.135}) & (\num{0.123})\\
Control & yes & yes & yes & yes & yes & yes\\
\midrule\\
Num.Obs. & \num{43} & \num{43} & \num{43} & \num{43} & \num{43} & \num{43}\\
R2 & \num{0.458} & \num{0.461} & \num{0.476} & \num{0.514} & \num{0.510} & \num{0.528}\\
R2 Adj. & \num{0.350} & \num{0.353} & \num{0.371} & \num{0.400} & \num{0.395} & \num{0.416}\\
\bottomrule
\multicolumn{7}{l}{\rule{0pt}{1em}+ p $<$ 0.1, * p $<$ 0.05, ** p $<$ 0.01, *** p $<$ 0.001}\\
\end{tabular}
\end{table}

Table \ref{tab:reg_eig} evaluates the effects of these adjustments on regression estimates. Columns (1) and (4) utilize the unadjusted eigenvector centrality (\textit{c\_eig}), Columns (2) and (5) use the corrected version (\textit{c\_eig\_corrected}), where zero-degree nodes are excluded, and Columns (3) and (6) employ the normalized measure (\textit{c\_eig\_normalized}), ensuring independence from network size.

The results reveal significant distortions when using unadjusted measures. In Columns (1) and (4), the negative correlation between \textit{c\_eig} and network size diminishes the magnitude and statistical significance of the \textit{log10(num\_hh)} coefficients. Similarly, \textit{c\_eig\_corrected} in Columns (2) and (5) yields comparable biases, indicating the sensitivity of both measures to network size (See also Table \ref{tab:alt_centrality}).

By contrast, the normalized eigenvector centrality (\textit{c\_eig\_normalized}) in Columns (3) and (6) produces stable and statistically significant coefficients for both \textit{c\_eig\_normalized} and \textit{log10(num\_hh)}. The consistency between these estimates demonstrates the robustness of the normalized centrality measure, making it suitable for cross-network comparisons. Additionally, models employing \textit{c\_eig\_normalized} achieve the highest adjusted $R^2$ values, further emphasizing its value in generating reliable and interpretable results.

These findings highlight the importance of normalization and correction in eigenvector centrality to ensure meaningful analysis across networks of varying sizes and structures.

\end{document}